\documentclass[12pt,a4paper,english,reqno]{amsart}
\usepackage[T1]{fontenc}
\usepackage[latin9]{inputenc}
\usepackage{amsthm}
\usepackage{amssymb}
\usepackage{tikz}
\usetikzlibrary{arrows,chains,matrix,positioning,scopes}

\makeatletter
\tikzset{join/.code=\tikzset{after node path={%
\ifx\tikzchainprevious\pgfutil@empty\else(\tikzchainprevious)%
edge[every join]#1(\tikzchaincurrent)\fi}}}
\makeatother

\tikzset{>=stealth',every on chain/.append style={join},
         every join/.style={->}}

\makeatletter

\pdfpageheight\paperheight
\pdfpagewidth\paperwidth

\numberwithin{equation}{section}
\numberwithin{figure}{section}

\usepackage{amsfonts}

\usepackage{amsthm}     
\usepackage{amscd}      
\usepackage{euscript}
\usepackage[matrix,arrow,curve]{xy}
\usepackage{pb-diagram}

{
       \newtheorem{theorem}{Theorem}[section]
       \newtheorem{proposition}[theorem]{Proposition}
       \newtheorem{lemma}[theorem]{Lemma}

       \newtheorem{conjecture}[theorem]{Conjecture}

{\theoremstyle{definition}

       \newtheorem{definition}[theorem]{Definition}
       
}

\newcommand{\RR}{{\mathbb{R}}}

\newcommand{\pk}{\mathbb{CP}^{k-1}}
\newcommand{\Oone}{\mathcal{O}_{\mathbb{CP}^{k-1}}(1)}
\newcommand{\wpderi}{\frac{\p}{\p w_i}|_{\tilde{\eta}}}
\newcommand{\wpder}{\frac{\p}{\p w_i}}
\newcommand{\half}{\frac{1}{2}}
\newcommand{\xs}{x_s^{\tilde{\eta}}}
\newcommand{\ys}{y_s^{\tilde{\eta}}}

\newcommand{\W}[2]{{H^{#1,#2}}}

\newcommand{\Wnorm}[3]{\left\Vert{#1}\right\Vert_{H^{#2,#3}}}

\newcommand{\Wnormloc}[4]{\left\Vert{#1}\right\Vert_{H^{#2,#3}(#4)}}

\newcommand{\Lnorm}[2]{\left\Vert{#1}\right\Vert_{L^#2}}
\newcommand{\Lnormloc}[3]{\left\Vert{#1}\right\Vert_{L^#2(#3)}}
\newcommand{\super}[1]{\varphi_{+,#1}}
\newcommand{\sub}[1]{\varphi_{-,#1}}
\newcommand{\p}{\partial}

\newcommand{\pderi}{\frac{\p}{\p t}|_{t=0}}

\newcommand{\M}{\{I\}\cup M^l}
\newcommand{\R}{\{R\}\cup M^R}
\newcommand{\Hol}{Hol_r\left(\Sigma,\mathbb{CP}^{k-1}\right)}
\newcommand{\pRi}{\p^R|_{(\tilde{\eta},z)}}

\makeatother

\usepackage{babel}

\usepackage{babel}
\begin{document}

\title[Dynamics of Ab. Vortices Without Common Zeros in the Adia. Limit]{Dynamics of Abelian Vortices Without Common Zeros in the Adiabatic Limit}
\author{Chih-Chung Liu}
\maketitle

\begin{abstract}

On a smooth line bundle $L$ over a compact K\"ahler Riemann surface $\Sigma$, we study the family of vortex equations with a parameter $s$. For each $s \in [1,\infty]$, we invoke techniques in \cite{Br} by turning the $s$-vortex equation into an $s$-dependent elliptic partial differential equation, studied in \cite{kw}, providing an explicit moduli space description of the space of gauge classes of solutions. We are particularly interested in the bijective correspondence between the open subset of vortices without common zeros and the space of holomorphic maps. For each $s$, the correspondence is uniquely determined by a smooth function $u_s$ on $\Sigma$, and we confirm its convergent behaviors as $s \to \infty$. Our results prove a conjecture posed by Baptista in \cite{Ba}, stating that the $s$-dependent correspondence is an isometry between the open subsets when $s=\infty$, with $L^2$ metrics appropriately defined.

\end{abstract}

\section{Introduction}
The vortex equations, a set of gauge invariant equations characterizing the minimum of certain energy functionals on
 a Hermitian vector bundle, have been studied quite extensively. An early occurrence can be found in Ginzburg and
 Landau's description of the free energy of superconducting materials, which depends on the external electromagnetic
 potential and the state function of certain electron pairs known as the "Cooper pairs". Finding the equilibrium
 state of the material amounts to minimizing the free energy. See \cite{JT} for a complete description.

Various forms of the energy functionals are available in the literature. We shall vaguely refer to them as the
Yang-Mills-Higgs functional, with historical origins from the classical Yang-Mills functional on field strength
of electromagnetic waves. We will investigate a particular functional, which we describe below.

Let $L$ be a degree $r$ line bundle over an $n$-dimensional closed (compact with empty boundary) K\"ahler manifold $(M,\omega)$. Let $H$ be a
Hermitian metric on $L$ and let $\mathcal{A}(H)$ be the space of connections which are $H$-unitary.
Let $\mathcal{G}$ be the $H$-unitary gauge group of the bundle $L$. To fix the notations uniformly, we will replace the base manifold $M$ by $\Sigma$ if it is a closed Riemann surface of genus $b$, for which we make addition assumption that $r>2b-2$ so that the vector space of holomorphic sections $H^0(\Sigma, L)$ is of uniform dimension on $\Sigma$.

The Hermitian structure $H$ naturally defines $L^2$ norms, which induce corresponding norm topologies, on complex and vector valued forms. $\mathcal{A}(H)$ and $\Omega^0(L)$ (the space of smooth global sections of $L \to \Sigma$) in fact possess standard K\"ahler structures (see \cite{G} for details). With these preliminary structures, we consider
the parameterized Yang-Mills-Higgs functional defined on the space of $H$-unitary connections and $k$ tuples of
smooth sections

\[
YMH_{\tau.s}:\mathcal{A}(H)\times\Omega^{0}(L) \times \ldots \times \Omega^{0}(L)\to\RR,
\]
 given by:
\begin{equation}
YMH_{\tau.s}(D,\phi):=\frac{1}{s^{2}}\left|\left|F_{D}\right|\right|_{L^{2}}^{2}+\sum_{i=1}^k\left|\left|D\phi_i\right|\right|_{L^{2}}^{2}+\frac{s^{2}}{4}\left|\left|\sum_{i=1}^k|\phi_i|^2_H-\tau \right|\right|_{L^{2}}^{2}.\label{YMH}
\end{equation}
\noindent Here $F_{D}\in\Omega^{1,1}(M,End(L))\simeq \Omega^{1,1}(M)$ is the curvature of the connection
$D$ and $\phi=(\phi_1,\ldots,\phi_k)$ is understood to be a $k$ tuple of sections. The positive real constant $s$ possesses physical significance in various situations.
Mathematically, when $n=1$, the parameter $s$ in the functional represents how the Yang-Mills-Higgs functional
changes when deforming the metric by rescaling, that is, $\omega_s=s^2\omega$. The other parameter $\tau$ first
appears in \cite{Br}, in which $s=1$.

Applying standard K\"ahler identities, one can obtain the minimizing equations for $YMH_{\tau,s}$ (See \cite{Ba}
and \cite{Br} for derivations when $s=1$), referred to as the $s$-vortex equations:

\begin{equation}
\begin{cases}
F_{D}^{(0,2)}=0\\
D^{(0,1)}\phi=0\\
\sqrt{-1}\Lambda F_{D}+\frac{s^{2}}{2}(\sum_{i=1}^k|\phi_i|^2_H-\tau)=0.
\end{cases}\label{s-vortex}
\end{equation}

\noindent Here $(p,q)$ refers to the decomposition of forms with respect to a fixed complex structure of $M$. Recall that $\Lambda$ is the $L^2$ adjoint of the Lefschetz operator

\[L(\gamma)=\gamma \wedge \omega.\]

\noindent On $(1,1)$ forms, $\Lambda$ is simply the trace with respect $\omega$:

\[\Lambda(\gamma)=<\gamma,\omega>_\omega \in C^\infty(M).\]

\noindent The
first equation in \eqref{s-vortex} says that $D^{(0,1)}$ is integrable, hence that
it induces a holomorphic structure on $L$ (by a celebrated theorem of Newlander-Nirenberg). For $M=\Sigma$, this condition is automatic. The second equation says that each section $\phi_i$ is holomorphic with respect to this holomorphic structure, and we will adhere to this notational convention throughout this paper. The third equation imposes a relation between curvature forms and norms of the $k$ sections. In some literature, the first two equations are assumed and the third equations is called the vortex equation. We however, study the three equations altogether.

One of the main goals of this paper is to analyze the adiabatic limit of solutions to \eqref{s-vortex} $s \to \infty$. Formally, as $s$ increases,
the curvature term in the third equation in \eqref{s-vortex} becomes negligible. Therefore, it is reasonable to
define the vortex equation at $s=\infty$ to be:

\begin{equation}
\begin{cases}
F_{D}^{(0,2)}=0\\
D^{(0,1)}\phi=0\\
\sum_{i=1}^k|\phi_i|^2_H-\tau=0.
\end{cases}\label{infinity-vortex}
\end{equation}

\noindent The solutions to these equations are then pairs of integrable connections, and corresponding $k$ tuple of holomorphic sections with image lying in the sphere of radius $\tau$ (with respect to the norm defined by $H$). The systems in equations \eqref{s-vortex} and \eqref{infinity-vortex} differ by
the third equation and our focus is to understand the limiting behaviors of the solutions of the third equation
in \eqref{s-vortex} as $s \to \infty$. We will achieve this by first reducing the equation, as in \cite{Br}, to a
scalar non-linear PDE and then by successively approximating, as in \cite{kw}, these equations by means of linear
ones.

The invariance of equations \eqref{s-vortex} and \eqref{infinity-vortex} under natural $\mathcal{G}$ action allows us to define the space
of gauge classes of solutions:

\begin{definition}
For each $k$, $s$ and $\tau$, we define the moduli space of solutions \[\nu_k(s,\tau) =\{(D,\phi)\in \mathcal{A}(H) \times \Omega^0(L) \times \ldots \times \Omega^0(L)  \hspace{0.1cm}| \,\eqref{s-vortex} \hspace{0.1cm} \text{holds} \}/\mathcal{G}. \]
Also, we define \[\nu_k(\infty,\tau) =\{(D,\phi)\in \mathcal{A}(H) \times \Omega^0(L) \times \ldots \times \Omega^0(L)  \hspace{0.1cm}| \,\eqref{infinity-vortex} \hspace{0.1cm} \text{holds} \}/\mathcal{G}.\]
\label{first definition}
\end{definition}

The spaces of solutions to \eqref{s-vortex} and \eqref{infinity-vortex} are smooth (actually K\"ahler) manifold as they can be realized as the level set of certain moment map (see, for example, \cite{G} for details). Furthermore, smooth connections on line bundles are clearly irreducible, and $\mathcal{G}$ acts on sections and connections by multiplication and conjugation, respectively. It is therefore a free action,  making the quotient space $\nu_k(s,\tau)$ of solution spaces to \eqref{s-vortex} and \eqref{infinity-vortex} smooth manifolds (see Chapter 4 and 5 of \cite{DK} for detailed discussions). We will see, in section 2, that they are of finite dimensions.

Bradlow \cite{Br} \cite{Br1}, Garcia-Prada \cite{G} and Bertram et.al \cite{BDW} have described $\nu_k(1,\tau)$
quite thoroughly for $M=\Sigma$. In fact, we will see that for finite values of $s$ and $\tau$ large enough, $\nu_k(s,\tau)$ are
all topologically identical.

Before we state the main statements, we pause briefly to examine the two real parameters $s$ and $\tau$ in the
vortex equations \eqref{s-vortex}. One notes that the gauge class $[D,\phi]$ satisfies \eqref{s-vortex} with $s$ and $\tau$
precisely when $[D,\frac{\phi}{\sqrt{\tau}}]$ does, with $s$ and $\tau$ replaced by $s\sqrt{\tau}$ and 1,
respectively. That is, the rescaling

\[[D,\phi] \mapsto [D,\frac{\phi}{\sqrt{\tau}}]\]

\noindent defines a bijection between $\nu_k(s,\tau)$ and $\nu_k(s\sqrt{\tau},1)$. These two parameters can
therefore be combined into one without altering the descriptions of the solution spaces. However, for the
convenience of comparing with classical results, we keep them separated, with the understanding that they are not
independent parameters.

Our main result is motivated by results in \cite{Ba} and \cite{BDW}. We are interested in the subset of
$\nu_k(s,\tau)$ consisting of $k$ sections without common zeros:

\begin{definition}
\[
\nu_{k,0}(s,\tau)=\{[D,(\phi_1, \cdots, \phi_k)] \in \nu_k(s,\tau)\hspace{0.1cm}|\hspace{0.1cm}\cap_{i=1}^{k}\phi^{-1}_i(0)=\emptyset\}.
\]
\label{introduction non vanishing sections}
\end{definition}

This subset is open and dense with respect to the quotient topology of $\nu_{k,0}(s,\tau)$ descended from the norm topology of the solution space. Indeed, the $\mathcal{G}$-equivariant evaluation map

\[ev([D,\phi,p]) := \phi(p) \in \mathbb{C}^k\]

\noindent is continuous with respect to the standard topology of $\mathbb{C}^k$ and $\nu_{k,0}(s,\tau)$ is the complement of $ev^{-1}(0,\ldots,0)$, which is closed by continuity. The density also follows obviously since the Hermitian structure is locally given by smooth functions and zeros can always be smoothly perturbed.

The topological descriptions of $\nu_{k,0}(s,\tau)$ have been studied extensively. Some references include \cite{CGRS}, \cite{MP}, \cite{Wi}, and \cite{Z}. For $M=\Sigma$, the spaces $\nu_{k,0}(s,\tau)$ are completely described in \cite{BDW} and \cite{Ba}. For $s,\tau$
large enough, there is a diffeomorphism

\[\Phi_s:Hol_r(\Sigma,\mathbb{CP}^{k-1}) \rightarrow \nu_{k,0}(s,\tau),\]

\noindent where $Hol_r(\Sigma,\mathbb{CP}^{k-1})$ is the space of degree $r$ holomorphic maps from $\Sigma$
to $\mathbb{CP}^{k-1}$. (Recall that $r$ is the topological degree of the line bundle $L$). Of course, the smooth structure of $\Hol$ needs to be specified and a brief summary of relevant classical descriptions will be provided in section 2. The constructions of diffeomorphisms $\Phi_s$'s will also be provided there.

In section 4, we establish metrical relations between these two spaces and their dependence on the parameter $s$. These results are applications of our main analytic result, showing that the family $\Phi_s$ can be very well controlled. More precisely, we will see, in section 2 and 3, that $\Phi_s$ identifies a holomorphic map from $\Sigma$ to $\pk$ with a vortex $[D,\phi]\in \nu_{k,0}(s,\tau)$ via a complex gauge element in $\mathcal{G}_{\mathbb{C}}$. On a line bundle, such an element is uniquely determined by a real smooth function $u_s$ on $\Sigma$, and we show that they
exhibit convergent behaviors as $s \to \infty$. The analytic result is of independent interest. Let $H^{l,p}$ denote the Sobolev $l,p$ space on a compact Riemannian manifold $M$. Presented as the Main Theorem in section 3, the result is:

\begin{theorem} [Main Theorem]

On a compact Riemannian manifold $M$ without boundary, let $c_1$ be any constant, $c_2$ any positive constant,
and $h$ any negative smooth function. Let $c(s) = c_1 -c_2 s^2$, for each $s$ large enough,
the unique solutions $\varphi_s \in C^\infty(M)$ for the equations

\[
\Delta\varphi_{s}=c(s)-s^2he^{\varphi_s}
\]

\noindent are uniformly bounded in $\W{l}{p}$ for all $l \in \mathbb{N}$ and $p \in [1,\infty]$. Moreover, in the limit
$s \to \infty$, $\varphi_{s}$ converges smoothly (i.e. uniformly in all $H^{l,p}$) to

\[\varphi_\infty = \log\left(\frac{c_2}{-h}\right),\]

\noindent the unique solution to

\[
he^{\varphi_{\infty}}+c_{2}=0.
\]

\label{Main Theorem Intro}
\end{theorem}

This result aids us in the study of dynamics of vortices, or evolutions of metrics, first explored by
Manton (\cite{M}). There, an approximating model governed by the geodesics of a naturally defined
$L^2$ metric (or kinetic energy) on $\nu_k(1,\tau)$ is provided for the motion of vortices. This motivated a
need for descriptions of the natural $L^2$ metric in precise mathematical languages.
(See, for example, \cite{S} and \cite{Ro}.)
A more concrete description is available when $k=1$, when $\nu_1(1,\tau)$ is identified with a familiar space with
explicit coordinates.
Samols has provided a semi-explicit coordinate expression of the natural $L^2$ metric using the coordinates of the
parametrizing space. It is natural to consider what happens to the metrics as one varies the parameters $s$, $k$,
and $\tau$, and let $s$ approach infinity.
Baptista has proposed a conjecture in \cite{Ba}, asserting that the $s$-dependent $L^2$ metrics on the open subset
$\nu_{k,0}(s,\tau)$ can be pulled back to a metric on $Hol_r(\Sigma,\mathbb{CP}^{k-1})$. As $s \to \infty$, it was
conjectured that the pullback metric approaches a familiar one on $Hol_r(\Sigma,\mathbb{CP}^{k-1})$. In section 4,
we apply the Main Theorem \ref{Main Theorem} to prove this conjecture.

It is worthwhile to point out that the convergent behaviors of vortices on $\nu_{k,0}(s,\tau)$ have been
established elsewhere. In \cite{Z}, the compactness properties of vortices with uniformly bounded energies have been
thoroughly described for the more general case of symplectic vortex equations. The convergent discussions for
our particular setting have appeared in \cite{X}. The novelty of our work lies in the scrutiny of the limiting
elements in a precise analytic framework using rather elementary techniques, and the fact that our results are a
consequence of a more general theorem on the uniform regularity of solutions to a family of semilinear P.D.E. on a
general closed Riemannian manifold. The other novelty is its application toward a precise formulation and rigorous proof of Baptista's conjecture (Conjecture 5.2 in \cite{Ba})
on the dynamics of vortices, for which other established results do not seem immediately applicable.

\section{Backgrounds and Statements of the Results}
We begin by briefly summarizing the descriptions of $\nu_k(s,\tau)$. Readers familiar with constructions in \cite{Br}
and \cite{BDW} may skip to  Lemma \ref{BDW map}. One must first ensure the conditions for existences of the solutions to
the vortex equations \eqref{s-vortex} and \eqref{infinity-vortex}, or, equivalently, the non-emptiness of $\nu_k(s,\tau)$.
For a vector bundle of general rank, the non-emptiness is equivalent to a $\tau$ and $\phi$ dependent algebraic
properties on subsheaves of $E$ called $\tau$-stability. See \cite{Br1} and \cite{BDW} for detailed explanations.
Throughout this paper, we restrict our attention to rank $1$ vector bundles, or line bundles denoted by $L$.
Having no nontrivial proper subsheaf, the $\tau$-stability degenerates to a condition solely on $\tau$. By integrating
the third equation in \eqref{s-vortex}, a necessary condition for solution to exist is that

\[s^2 \tau \geq \frac{4 \pi r}{\text{vol} M}.\]

\noindent We will see that it is also sufficient. In the case $s=1$, $k=1$, and $M=\Sigma$, a Riemann surface, we have:

\[\nu_1(1,\tau)=
\begin{cases}
\emptyset &;\tau < \frac{4 \pi r}{\text{vol}\Sigma } \\
Jac^r \Sigma &; \tau = \frac{4 \pi r}{\text{vol}\Sigma } \\
Sym^r \Sigma &; \tau > \frac{4 \pi r}{\text{vol}\Sigma }, \\
\end{cases} \label{moduli spaces}\]

\noindent where $r=deg(L)$. Here, $Sym^r\Sigma$ is the space of unordered $r$ tuple of points of $\Sigma$ (or the space of divisors of $\Sigma$ with degree $r$)
and $Jac^r \Sigma$ is the Jacobian torus of $\Sigma$ parametrizing holomorphic structures of $L$. (See \cite{Br}).

The parameter $s$ does not alter the conclusion. We have seen that the effect of $s^2$ can be thought of as scaling the section $\phi$ and replacing $\tau$ by $s^2\tau$.
 This observation generalizes Bradlow's result in \cite{Br} naturally:

\begin{equation}
\nu_1(s,\tau)=
\begin{cases}
\emptyset &;s^2\tau < \frac{4 \pi r}{\text{vol}\Sigma } \\
Jac^r \Sigma &; s^2\tau = \frac{4 \pi r}{\text{vol}\Sigma } \\
Sym^r \Sigma &; s^2\tau > \frac{4 \pi r}{\text{vol}\Sigma }. \\
\end{cases} \label{s moduli spaces}
\end{equation}

The crucial step to achieve these descriptions is to switch perspective, from one in which we look for pairs $(D,\phi)$ on a
bundle with fixed unitary structure, to one in which we look for a metric on a fixed holomorphic line bundle with a
prescribed holomorphic section. In the second perspective, the analytic tools from \cite{kw} can be applied to solve for
the special metrics. The equivalence of the two perspectives is given in \cite{Br}, and we briefly summarize them here.

Let $\mathcal{C}$ be the space of holomorphic structures of $L$, that is, the collection of
$\mathbb{C}$-linear operators

\[\bar{\partial}_L: \Omega^0(L) \rightarrow \Omega^{0,1}(L)\]

\noindent satisfying the Leibiniz rule. It is a classical fact from differential geometry that given a Hermitian structure
$H$, we have $\mathcal{A}(H) \simeq \mathcal{C}$. The original approach toward solving vortex equations is to fix a
Hermitian structure $H$ and consider the following space:

\[\mathcal{N}_k:=\{(D,\phi) \in \mathcal{A}(H) \times \Omega^0(L) \times \ldots \times \Omega^0(L) \,\,| \,\,D^{(0,1)}\phi_i=0\,\, \forall i\}.\]

\noindent For a fixed $H$ this space is bijective to

\begin{equation}
\{(\bar{\partial},\phi) \in \mathcal{C} \times \Omega^0(L) \times \ldots \times \Omega^0(L) \,\,| \,\,\bar{\partial}\phi_i=0\,\, \forall i\}.
\label{second description}
\end{equation}

\noindent We then aim  to find a pair in $\mathcal{N}_k$ so that the third equation of the vortex
equations \eqref{s-vortex} is satisfied. The solvability statement we seek is:

\vspace{0.5cm}
\begin{center}
\emph{Given a Hermitian structure $H$, we find all pairs $(D,\phi)\in\mathcal{N}_k$ that solve the third equation
of \eqref{s-vortex}.}
\end{center}
\vspace{0.5cm}

Alternatively, we may start without fixing the Hermitian structure. The second description of
$\mathcal{N}_k$ \eqref{second description} above continues to make sense, and we pick an arbitrary pair
$(\bar{\partial},\phi) \in \mathcal{N}_k$. This pair determines a unique connection, and thus a unique curvature,
once a Hermitian metric $K$ is chosen. We specifically choose $K$ so that the third equation of \eqref{s-vortex} is
satisfied with this metric, and the curvature it defines:

\[\sqrt{-1}\Lambda F_{K} + \frac{s^{2}}{2}(\sum_{i=1}^k|\phi_i|^2_K-\tau)=0.\]

\noindent Here, $F_K$ is the curvature of the unique $K$-unitary connection with holomorphic structure $\bar{\p}$.
To spell out the details, the alternative approach of the problem requires us to start with the space

\[\mathfrak{T}_k=\{(\bar{\partial},\phi,K)\in \mathcal{C} \times \Omega^0(L) \times \ldots \times \Omega^0(L) \times \mathcal{H}\},\]

\noindent where $\mathcal{H}$ is the space of Hermitian structures of $L$. We fix the first two components,
and the solvability statement states the unique existence of the corresponding third component:

\vspace{0.5cm}
\begin{center}
\emph{Given a pair $(\bar{\partial},\phi) \in \mathcal{C} \times \Omega^0(L) \times \ldots \times \Omega^0(L)$ such that $\bar{\p} \phi_i=0\;\forall i$, we find all Hermitian metrics $K$ solving the third equation of \eqref{s-vortex} with the curvature and norms determined by $K$.}
\end{center}
\vspace{0.5cm}

Such an approach allows us to apply analytic techniques to solve the vortex equations. It is well known that any two
Hermitian metrics are related by a positive, self-adjoint
bundle endomorphism , i.e. by an element in the complex gauge group $\mathcal{G}_\mathbb{C}$. On a line bundle $L$,
$End(L)\simeq L\otimes L^{*}\simeq\mathcal{O}_{M}$,
so any two $C^\infty$-Hermitian metrics
on $L$, say $H$ and $K$, are related by $K=f\, H$ with $f\in C^\infty(M)$ and $f=e^{2u}>0$
for some $u \in C^\infty(M)$. Therefore, starting with a background metric $H$, finding the special metric $K$ is
equivalent to finding the unique function $u$ satisfying a certain elliptic PDE determined by the third equation
of \eqref{s-vortex}.

This alternative approach is equivalent to the original one only if we are able to build a bijection between the
two solution spaces, up to gauges. The gauge group for the alternative space is however not only $\mathcal{G}$ but
rather $\mathcal{G}_{\mathbb{C}}$, the complex gauge group. It acts on $\mathfrak{T}_k$ by

\begin{equation}
g^*(\overline{\partial}_{L},\phi,H)=(g \circ\overline{\partial}_{L}\circ g^{-1},\phi g,Hh).
\label{gauge action}
\end{equation}

\noindent Here, $h=g^*g=e^{2u}$ for a smooth real function $u$. Unlike the unitary gauge $\mathcal{G}$, this action
does not necessarily preserve the $H$-norm of $\phi$. We define

\begin{equation}
\mathcal{T}_k(s,\tau)=\{(\bar{\partial}_{L},\phi,K)\in \mathfrak{T}_k \hspace{0.1cm};\hspace{0.1cm}\eqref{s-vortex} \text{ holds with metric } K\} /\mathcal{G}_\mathbb{C}.
\label{moduli space another form}
\end{equation}

\noindent We now summarize the bijection between $\mathcal{T}_k(s,\tau)$ and $\nu_k(s,\tau)$. The proof is directly
reproduced from Proposition 3.7 in \cite{Br}, proved for $k,s=1$. However, it is by no means special to that particular value, and the
proof applies to general values of $k,s$ without any modification.

\begin{lemma}\cite{Br}
There is a bijective correspondence between $\nu_k(s,\tau)$ and $\mathcal{T}_k(s,\tau)$.
\label{identification of moduli space}
\end{lemma}

\begin{proof} \emph{(Sketch)}
To define the forward map $P_s:\nu_k(s,\tau) \to \mathcal{T}_k(s,\tau)$, we take $[D,\phi] \in \nu_k(s,\tau)$.
The integrability of $D$ implies that its anti-holomorphic part $D^{(0,1)}$ defines a holomorphic structure, and
we define

\[P_s([D,\phi])= [D^{(0,1)},\phi ,H],\]

\noindent where $H$ is the background metric for which $D$ is $H$-unitary. For the inverse map $G_s$, take
$[\bar{\partial}_{L},\phi,K]\in\mathcal{T}_k(s,\tau)$. The Hermitian metric $K$ on $L$ is related to $H$ by $K=e^{2u}H$,
and $g=e^u$ acts on holomorphic
structure and sections as in \eqref{gauge action}. We define

\[G_s([\bar{\partial}_{L},\phi,K])=[D(g^*\bar{\partial}_{L},H),\phi\circ g],\]

\noindent where $D(g^*\bar{\partial}_{L},H)$ is the metric connection of $H$ with holomorphic
structure $g^*\bar{\partial}_{L}$. That the pair $(D(g^*\bar{\partial}_{L},H),\phi\circ g)$ solves the vortex
equation \eqref{s-vortex} and that $P_s$ and $G_s$ are inverse to each other are proved in \cite{Br}.
\end{proof}

The alternative perspective yields a much more intuitive understanding of Bradlow's description of $\nu_1(1,\tau)$
for large $\tau$. An element $<z_1,\ldots,z_r>\in Sym^r\Sigma$ uniquely determines a pair $(\bar{\partial},\phi)$
with $\bar{\partial}\phi = 0$, up to $\mathcal{G}_{\mathbb{C}}$ action, that vanishes precisely at these points.
The identification

\[\mathcal{T}_1(1,\tau) \simeq Sym^r \Sigma\]

\noindent is achieved once we ensure that the third component $K$ is uniquely determined by the first two, up to
$\mathcal{G}_{\mathbb{C}}$.

With the identification in Lemma \ref{identification of moduli space}, finding $(D,\phi)$
to satisfy equation \eqref{s-vortex} is equivalent to fixing a holomorphic structure $\bar{\partial}_L$, a
holomorphic section $\phi$, and finding a special metric $K_s=He^{2u_s}$ so that equation \eqref{s-vortex} is satisfied
 with this metric. As we have claimed, this turns the third equation in \eqref{s-vortex}, which is a tensorial one,
  into a scalar equation of $u_s$. Moreover, it turns the question of understanding the limiting behaviors of vortices
  into analyzing the convergent behaviors of $u_s$.

Before we describe $\nu_k(s,\tau)$ for general $k$, we observe that near the adiabatic limit $s=\infty$, the third possibility in \eqref{s moduli spaces}
prevails. As we are mainly interested in the asymptotic behaviors of vortices, that possibility will be the focus
of our attention, and $\tau$ dependence becomes insignificant. We will therefore assume $\tau=1$ and write $\nu_k(s)$
instead of $\nu_k(s,1)$ from now on.

\[\nu_k(s) := \nu_k(s,1),\]

\noindent and

\[\nu_{k,0}(s) := \nu_{k,0}(s,1)\]

\noindent for large values of $s$.

The generalized description to \eqref{s moduli spaces} is given in \cite{BDW}. We are particularly interested in
the open subset $\nu_{k,0}(s)$ of $\nu_k(s)$ defined in Definition \ref{introduction non vanishing sections}.
Let $\mathcal{T}_{k,0}(s)$ be the corresponding open subset of $\mathcal{T}_k(s)$ via the identification in
Lemma \ref{identification of moduli space}. It is obvious that $\nu_{k,0}(\infty)=\nu_k(\infty)$ since the third
equation of \eqref{infinity-vortex} prohibits simultaneous vanishing of the $k$ sections.
It is also clear that $\nu_{1,0}(s)$ is empty for all $s<\infty$, since any global holomorphic section of a line
bundle with degree $r$ must vanishes exactly at $r$ points, counting multiplicities. This is not the case when we
have more than one section. In fact, it has been shown in \cite{BDW} that

\begin{equation}
 Hol_r(\Sigma,\mathbb{CP}^{k-1})\simeq \nu_{k,0}(1),
\label{BDW identification}
\end{equation}

\noindent where the equivalence above is in fact a diffeomorphism, under the initial assumption $r>2b-2$. We hereby provide a brief description of the manifold structure of $\Hol$ in this circumstance. It is a classical fact that for $r>2b-2$, $\Hol$ is a smooth manifold of complex dimension

\[m=kr-(k-1)(b-1).\]

\noindent Every $f \in \Hol$ is of the form

\[f=[f_1,\ldots,f_k],\]

\noindent where each $f_j$ is a meromorphic function on $\Sigma$. Since $f$ is of degree $r$, each $f_j$ vanishes exactly on a divisor $E_j \in Sym^r(\Sigma)$. The space $Sym^r(\Sigma)$ is locally diffeomorphic to $\mathbb{C}^r$, by identifying an unordered $r$-tuple $<z_1,\ldots,z_r>$ with the coefficients of the monic polynomial $(z-z_1)\cdots(z-z_k)$. Each $f \in \Hol$ is then associated with an element in $Sym^r(\Sigma)\times \ldots \times Sym^r(\Sigma)$, a complex manifold of dimension $kr$. Clearly, not every $(E_1,\ldots,E_k) \in Sym^r(\Sigma)\times \ldots \times Sym^r(\Sigma)$ determines a holomorphic map. An immediate restriction is that

\begin{equation}
E_1 \bigcap \ldots \bigcap E_k = \emptyset,
\label{empty intersection}
\end{equation}

\noindent and therefore we restrict our attention to

\[Div_r^k := \{(E_1,\ldots,E_k) \in Sym^r(\Sigma)\times \ldots \times Sym^r(\Sigma) \;\;|\;\;\bigcap_{i=1}^k E_i = \emptyset\},\]

\noindent which is still of dimension $rk$ since $Div_r^k$ is clearly an open subset of $Sym^r(\Sigma)\times \ldots \times Sym^r(\Sigma)$. The only other condition for $(E_1,\ldots,E_k) \in Div_r^k$ to determine a unique holomorphic map is given in Corollary 1.10 in \cite{KM}. It requires that

\begin{equation}
\mu_r(E_1) = \ldots = \mu_r(E_k).
\label{condition for Abel Jacobi}
\end{equation}

\noindent The map $\mu_r$ is the generalized Abel-Jacobi map. Precisely, for $E_j = p_j^1 + \ldots + p_j^r$, we define

\begin{equation}
\mu_r(E_j) := \mu(p_j^1) + \ldots \mu(p_j^r),
\label{definition of Abel Jacobi}
\end{equation}

\noindent where $\mu : \Sigma \to \mathbb{C}^b/\mathbb{Z}^{2b} \simeq (\mathbb{S})^{2b}$ is the classical Abel-Jacobi map. For $r>2b-1$, the rank of the differential of $\mu_k$ is of rank $b-1$ (cf. Proposition V.4.7 of \cite{Gr}). The defining condition \eqref{condition for Abel Jacobi} for the space of $k$ tuples of divisors corresponding to $\Hol$ then consists of $k-1$ equations defined by maps with differentials of rank $b-1$. This correspondence therefore defines a manifold structure of $\Hol$, and \eqref{condition for Abel Jacobi} then reduces the original dimension $kr$ by $(k-1)(b-1)$.  More details can be found in section 1 of \cite{KM} and section 5.4 of \cite{Mi}.

With these preliminary knowledge recalled, we state

\begin{lemma} \cite{Ba}
For each $s \in [1,\infty]$, there is a diffeomorphism

\[\Phi_s: Hol_r(\Sigma,\mathbb{CP}^{k-1}) \to \nu_{k,0}(s),\]

\noindent in the smooth structures described immediately above (for $\Hol$) and section 1(for $\nu_{k,0}(s)$).

\label{BDW map}

\end{lemma}

\begin{proof} (\emph{Sketch})

We only sketch the outline of the construction of $\Phi_s$. Details and justifications are provided in section 3.

The inverse map $\Phi_s^{-1}$ is obvious. For $k$ sections $\phi = (\phi_1,\ldots,\phi_k)$ without common
 zeros, we can construct maps from $\Sigma$ to $\mathbb{CP}^{k-1}$ defined by

\begin{equation}
\Phi_s^{-1}([D,\phi])(z):=\tilde{\phi}(z)=[\phi_1(z),\ldots,\phi_k(z)].
\label{definition of Phi s inverse}
\end{equation}

\noindent The right hand side of \eqref{definition of Phi s inverse} is well defined, as $\phi_1(z),\ldots,\phi_k(z)$
are never zeros simultaneously. Moreover, on a $U(1)$ line bundle, the transition map multiplies each section by a
uniform nonzero scalar. Therefore, \eqref{definition of Phi s inverse} is a globally defined holomorphic map from
$\Sigma$ to $\mathbb{CP}^{k-1}$.

The construction of the forward map is also standard. We start with a holomorphic map $\tilde{\phi}\in \Hol$. Consider $\mathcal{O}_{\mathbb{CP}^{k-1}}(1)$, the anti-tautological line bundle over $\pk$ with hyperplane sections $s_1,\ldots,s_k$. Each $s_j$ vanishes precisely on the hyperplane defined by $z_j=0$.
Let $L=\tilde{\phi}^*\mathcal{O}_{\mathbb{CP}^{k-1}}(1)$ be the pullbacked line bundle on $\Sigma$ endowed with sections $\phi=(\phi_1,\ldots,\phi_k)\in \Omega^0(L)^{\oplus k}$ by pulling back
$s_1,\ldots,s_k$ via $\tilde{\phi}$. The map $\tilde{\phi}$ also endows a
holomorphic structure $\bar{\partial}_L$ and a background metric $H$ on $L$ when a background metric is given on $\mathcal{O}_{\mathbb{CP}^{k-1}}(1)$. The first part of section 3 is to modify Bradlow's arguments
 in \cite{Br} to look for a special metric $H_s$, related to $H$ by a gauge transformation $H_s=He^{2u_s}$, where $u_s$
 is a positive smooth function. The vortex equation \eqref{s-vortex} is to be satisfied if $H$ is replaced by $H_s$.
 The triplet $[\bar{\partial}_L,\phi,H_s]\in \mathcal{T}_{k,0}(s)$ corresponds via Bradlow's identification in
 Lemma \ref{identification of moduli space} to $[D_s,e^{u_s}\phi] \in \nu_{k,0}(s)$, where $D_s$ is the metric
 connection with respect to holomorphic structure $e^{u_s}\circ \bar{\partial}_L \circ e^{-u_s}$ and the Hermitian
 metric $H$, and we define

\begin{equation}
\Phi_s(\tilde{\phi})=[D_s,e^{u_s}\phi].
\label{definition of Phi s}
\end{equation}

Both $\Phi_s$ and $\Phi_s^{-1}$ are smooth in the smooth structures provided above. Indeed,  a holomorphic map from $\Sigma$ to $\pk$ is labeled by the $k$ tuple of divisors characterizing the zeros of each of its components. Perturbing the zeros smoothly results in smooth variation of corresponding pullback $k$ sections on $L \to \Sigma$, and therefore $\Phi_s$ is smooth. The smoothness of $\Phi_s^{-1}$ is obvious.

The unique existence of $u_s$ is guaranteed by the following theorem, which is proved with essentially identical reasonings from Lemma 4.1, Theorem 4.2, and Theorem 4.3 in \cite{Br}:

\begin{theorem}
[Existence and Uniqueness of $u_s$]

Fix $s^2 \in [\frac{4\pi r}{Vol \Sigma},\infty]$ and a Hermitian line bundle $(L,H)$ over $\Sigma$. Given a
holomorphic structure $\bar{\p}_L$ of $L$ and $k$ sections $\phi=(\phi_i)_i$ so that

\[\bar{\p}_L\phi_i=0 \;\;  \forall i,\]

\noindent there exists $u_s \in \mathcal{C}^\infty$ such that $[D_s,e^{u_s}\phi]\in
\left(\mathcal{A}(H)\times \Omega^0(L)^{\oplus k}\right)/\mathcal{G}$ solves the vortex equation \eqref{s-vortex} (i.e. $[D_s, e^{u_s}\phi] \in \nu_k(s)$).
\label{Identification of open subset}
\end{theorem}

The only generalization of Theorem \ref{Identification of open subset} from the particular cases in \cite{Br}
is the introduction of the parameter $s$ (in \cite{Br}, $s=1$). However, no significant modification of the original proofs is required. Nevertheless, to clarify the geometric and analytic role of $s$ as it approaches the asymptotic value, we must reproduce arguments from \cite{Br} and \cite{kw}. The proof to Theorem \ref{Identification of open subset} commences at the beginning of section 3.

Pending the proof of Theorem \ref{Identification of open subset} above and analytic details of the correspondence $\Phi_s$, to be presented in section 3, the sketch of the proof of Lemma \ref{BDW map} is now complete.

\end{proof}

The Main Theorem of this paper, Theorem \ref{Main Theorem}, further establishes significant controls of these real smooth functions $u_s$ uniquely determining the complex gauges.

\begin{theorem}
[Conclusion of the Main Theorem]
The functions $u_s$ converges to $u_\infty$ in all Sobolev spaces, and therefore smoothly as $s \to \infty$.
\end{theorem}

This Theorem proves a conjecture posed by Baptista in \cite{Ba} on dynamics of vortices. On $\nu_{k,0}(s)$, we
consider the natural $L^2$ metric given as follows. For each $(D_{s},\phi_{s}) \in \mathcal{A}(H) \times \Omega^0(L)^{\oplus k}$, we define

\begin{equation}
g_{s}((\dot{A_{s}},\dot{\phi_{s}}),(\dot{A_{s}},\dot{\phi_{s}}))=\int_{M}\frac{1}{2s^{2}}\dot{A_s}\wedge \bar{*}
\dot{A_s}+<\dot{\phi_{s}},\dot{\phi_{s}}>_{H}vol_M,
\label{L 2 metric}
\end{equation}

\noindent where $(\dot{A}_s,\dot{\phi}_s)$ is an element of $T_{(D_s,\phi_s)}\left(\mathcal{A}(H) \times \Omega^0(L)^{\oplus k}\right) \simeq \Omega^1(\Sigma)\times \Omega^0(L)^{\oplus k}$, chosen orthogonally to the gauge transformation. The second term of the
 integrand makes sense since the tangent space to sections is identified with itself, and we adopt the notation

\[<\phi,\psi>_H := \sum_{i=1}^k \left<\phi_i,\psi_i \right>_H.\]

\noindent Picking the tangent vectors in directions orthogonal to gauge transformations, $g_s$ descends to a metric on
$\nu_{k,0}(s)$, which is identified with  $Hol_r(\Sigma,\mathbb{CP}^{k-1})$ via $\Phi_s$. We then pull back
$g_s$ via $\Phi_s$ to a metric on $Hol_r(\Sigma,\mathbb{CP}^{k-1})$ and try to compare it with the ordinary $L^2$
metric of the space of holomorphic maps. Baptista's conjecture, a rather holistic statement, is stated roughly as follows:

\begin{conjecture} [Conjecture 5.2 in \cite{Ba}]
The pull back metrics $\Phi_s^*g_s$ converge pointwise to a multiple of the ordinary $L^2$ metric of the space of holomorphic maps.
\end{conjecture}

In section 4, we provide a precise formulation of this conjecture, as well as a precise convergent statement of the pullbacked metrics, before rigorously proving it.

\section{Main Constructions}
We first prove Theorem \ref{Identification of open subset}. To do so, we identify the equations for the unique gauges
$e^{u_s}$, that transform the initial data into solutions of vortex equations \eqref{s-vortex}. Before we begin, we state the well-known maximum principle for invertible elliptic operators, which is crucial in our analytic derivations.

\begin{lemma}
[Maximum Principle]
For the elliptic operator $L=\Delta -k$, where $k$ is any smooth positive function, the following is true: for
any $p > dim M$, if $u \in \W{2}{p}$ satisfies $Lu \geq 0$, then $u \leq 0$.
\label{Maximum Principle}
\end{lemma}

See, for example, (3.15) in \cite{kw} for the proof.

We now present the proof of Theorem \ref{Identification of open subset}.

\begin{proof} (\emph{of Theorem \ref{Identification of open subset}})
To begin, we
briefly summarize Bradlow's construction of the PDE's for $u_s$ to satisfy. If $H_s=e^{2u_s}H$ one has:
\[
\sqrt{-1}\Lambda F_{H_s}=\sqrt{-1}\Lambda F_{H}+\sqrt{-1}\Lambda\bar{\partial}(H^{-1}\partial{H}(e^{2u_s})).
\]
\noindent We get

\begin{equation}
\sqrt{-1}\Lambda F_{H_s}=\sqrt{-1}\Lambda F_{H}+2\sqrt{-1}\Lambda\bar{\partial}\partial u_s=\sqrt{-1}\Lambda F_{H}-\Delta_{\omega}u_s.\label{eq:3}
\end{equation}

\noindent Here, $\Delta_\omega$ is the Laplacian operator defined by K\"ahler
class $\omega$. Note that from standard K\"ahler identities, we have

\[2\sqrt{-1}\Lambda\bar{\partial}\partial u_s = - \Delta_\omega u_s,\]

\noindent where we use "analyst's Laplacian" here. It is defined so that on Euclidean $n$-space
$\omega=\sqrt{-1} \delta_{ij} dz^i \wedge d\bar{z}^j$,

\[\Delta_\omega f = \sum_{j=1}^n \frac{\partial^2 f}{\partial z^j \partial \bar{z}^j}.\]

\noindent We will omit the subscript $\omega$ from $\Delta_\omega$ if no confusion arises. Since $|\phi_i|_{H_s}^{2}=e^{2u_s}|\phi_i|_{H}^{2}\;\;\forall i,$
it follows that we can rewrite the last equation in \eqref{s-vortex}, with metric $H$ replaced by $H_s$,
as:
\begin{equation}
-\Delta u_s+\frac{s^{2}}{2}\,\sum_{i=1}^k|\phi_i|_{H}^{2}e^{2u_s}+\left(\sqrt{-1}\Lambda F_{H}-\frac{ s^{2}}{2}\right)=0.
\label{s-vortex 1}
\end{equation}

\noindent If we normalize the K\"ahler metric so that $Vol_{\omega}(M)=1$, we can define
\[
\begin{aligned}c(s): & =2\int_{\Sigma}\left(\sqrt{-1}\Lambda F_{H}-\frac{s^{2}}{2}\right)dvol_\omega=2\int_{\Sigma}\sqrt{-1}\Lambda F_{H}\omega^{n}-\frac{ s^{2}}{2}\, dvol_\omega\\
 & =2\int_{\Sigma}\sqrt{-1}\Lambda F_{H}\omega^{n}-\frac{ s^{2}}{2}=2c_1-\frac{ s^2}{2},
\end{aligned}
\]

\noindent where $c_{1}=\int_{\Sigma}\sqrt{-1}\Lambda F_{H}\omega^{n}$ is independent
of $s$ and $H$. Consider $\psi$, a solution to:

\begin{equation}
\Delta\psi=\left(\sqrt{-1}\Lambda F_{H}-\frac{s^{2}}{2}\right)-\frac{c(s)}{2}=\sqrt{-1}\Lambda F_{H}-c_{1},
\label{equaiton for psi}
\end{equation}

\noindent which is clearly independent of $s$.

\noindent Setting $\varphi_{s}:=2(u_{s}-\psi)$, $u_s$ satisfies \eqref{s-vortex 1} if and only if $\varphi_s$ satisfies:

\begin{equation}
\Delta\varphi_{s}-\frac{s^{2}}{2}(\sum_{i=1}^k\left|\phi_i\right|_{H}^{2}e^{2\psi})e^{\varphi_{s}}-c(s)=0.\label{Kazdan Warner Equation 0}
\end{equation}

\noindent This is of the form:

\begin{equation}
\Delta\varphi_{s}=-\left(\frac{s^{2}}{2}h\right)\, e^{\varphi_{s}}+c(s),\label{Kazdan Warner Equation}
\end{equation}

\noindent with $h =-\sum_{i=1}^k|\phi_i|_H^2 e^{2\psi}< 0$ and $c(s)<0$ (for large $s$). Proving Theorem \ref{Identification of open subset} then boils down to proving the unique existence of solutions to \eqref{Kazdan Warner Equation}. For $s=1$, Lemma 9.3 in \cite{kw} guarantees the unique solution to exist. We state the analogous theorem for general $s$ below. Our proof differs only slightly from that of \cite{kw} in which we choose certain required data more specifically to establish the uniformity and convergent behaviors of solutions $\varphi_s$ over $s$.  The theorem itself is of independent
interest, and applies to general functions on compact Riemannian manifold $(M,g)$. We however keep the notations
identical (except that we replace $\frac{h}{2}$ by $h$) for the convenience of application to our particular PDE \eqref{Kazdan Warner Equation}.

With these, we state our main constructions.

\begin{theorem}
[Existence and Uniqueness of $\varphi_s$]

On a compact Riemannian manifold $(M,g)$ without boundary, let $c_1$ be any constant, $c_2$ any positive constant,
and $h$ any negative smooth function. Let $c(s)=c_1-s^2 c_2$, the partial differential equation

\[\Delta \varphi_s =-s^2h e^{\varphi_s}+c(s)\]

\noindent has a unique smooth solution for all $s$ large enough.
\label{Kazdan Warner Analysis}
\end{theorem}

\noindent \begin{proof}

We first establish the uniqueness, a consequence of the maximum principle. Fix $s \in \mathbb{R}$, suppose
that $\varphi_s^1$ and $\varphi_s^2$ are smooth solutions to the equation and so is their difference $\varphi_s^1 - \varphi_s^2$. If $\varphi_s^1 \neq \varphi_s^2$, without loss of
generality, we may assume

\[\inf_M \{\varphi_s^1 (x) - \varphi_s^2 (x)\} < 0.\]

\noindent Since $M$ is compact and $\varphi_s^1-\varphi_s^2$ is smooth, the infimum must be attained at some
point $x_0 \in M$. We have

\[\varphi_s^1 (x_0) < \varphi_s^2 (x_0).\]

\noindent It follows that

\[\Delta (\varphi_s^1 - \varphi_s^2) (x_0) = -s^2h [e^{\varphi_s^1(x_0)} - e^{\varphi_s^2 (x_0)}] < 0,\]

\noindent since $-h>0$ and exponential functions are monotonically increasing. We have arrived at a contradiction
since the Laplacian of a smooth function has to be nonnegative at the point of minimum value. Therefore, the
solution for each $s$ has to be unique.

Following principles of the proof of Lemma 9.3 in \cite{kw}, we show the existence of solutions by constructing a sub-solution $\sub{s}$ satisfying

\[\Delta\sub{s} -c(s) +s^2h e^{\sub{s}} \geq 0,\]

\noindent and a super-solution $\super{s}$ satisfying

\[\Delta\super{s} -c(s) + s^2 h e^{\super{s}} \leq 0.\]

\noindent The two functions can be constructed independently of $s$ if the techniques from \cite{kw} are mimicked entirely. We however, choose pairs of super and sub
solutions that converge to the same function as $s \to \infty$. This construction will be useful in the Main Theorem
\ref{Main Theorem}, when we study the uniformity and convergent behaviors of $\varphi_s$.

Since the function $h$ is smooth and does not vanish, the function $ \log(-h) $ is smooth and therefore uniformly
bounded on the compact manifold $M$. Consequentially, there exists then a constant $K>0$ so that

\[\Delta (-\log(-h)) + K \geq 0,\]

and

\[\Delta (-\log(-h)) - K \leq 0.\]

For $s$ large enough so that $-K -c(s)\geq 0$, we define

\begin{equation}
\sub{s} = \log \left(\frac{-K-c(s)}{-s^2h}\right).
\label{sepcial sub solution}
\end{equation}

\noindent and
\begin{equation}
\super{s} = \log \left(\frac{K-c(s)}{-s^2h}\right).
\label{sepcial super solution}
\end{equation}

\noindent We have, for all $s$,

\[c(s) - s^2h e^{\sub{s}} = -K,\]

\noindent and

\[c(s) - s^2h e^{\super{s}} = K.\]

\noindent One can easily see that

\[\Delta \sub{s} = \Delta \super{s} = \Delta (-\log(-h)),\]

\noindent since $-\log(-h)$ is the only non constant part of their definitions on $M$. By our choice of $K$, we have

\[\Delta \sub{s} - c(s) + s^2 h e^{\sub{s}} \geq 0.\]

\noindent and

\[\Delta \super{s} - c(s) + s^2h e^{\super{s}} \leq 0,\]

\noindent verifying that they are indeed sub and super solutions.

The functions $\varphi_{+,s}$ and $\varphi_{-,s}$ are clearly uniformly bounded. In fact, one can readily verify that

\[\super{s} - \sub{s} = \log \left(\frac{K-c_1+c_2s^2}{-K-c_1+c_2s^2}\right) \to 0,\]

\noindent uniformly as $s \to \infty$.

We are now ready to solve the equation for each $s$. The solution will be the limit of certain iterative equations. Pick a constant $k>0$ so that

\[k \geq \sup_{s,M} -h e^{\varphi_{+,s}} ,\]

\noindent and consider the family of operators defined by

\[L_s = \Delta - s^2kI,\]

\noindent where $I: M \to M$ is the identity operator. Setting $\varphi_{0,s} := \varphi_{+,s}$. Since $s^2k>0$, $L_s$ is invertible for each $s$, and we can
therefore define the sequence $\{\varphi_{i,s}\}$ inductively by

\begin{equation}
\Delta \varphi_{i+1,s} -s^2k\varphi_{i+1,s} = c(s) - s^2k \varphi_{i,s} - s^2 h e^{\varphi_{i,s}}.
\label{iterative i s}
\end{equation}

\noindent That is, $\varphi_{0,s}=\varphi_{i+1,s}$ is the unique solution to the equation

\begin{equation}
L_s(f) =  c(s) - s^2k \varphi_{i,s} - s^2 h e^{\varphi_{i,s}}.
\label{Definition of L s}
\end{equation}

$\varphi_{+,s}$ is smooth by construction, and so is

\[c(s)-s^2k\varphi_{+,s}-s^2he^{\varphi_{+,s}}.\]

\noindent Schauder's estimate (cf. section 3 of \cite{kw}) on elliptic operators $L_s$ then ensures that $\varphi_{1,s}$ is smooth. By induction and the iterative relation \eqref{iterative i s} above, it follows that all $\varphi_{i,s}$ are smooth. A more crucial observation is that for all $i$
and $s$, we have the following monotonic and bounded-ness relation in $i$:

\begin{equation}
\varphi_{-,s} \leq \varphi_{i+1,s} \leq \varphi_{i,s} \leq \varphi_{+,s}
\label{monotinicity of iterative equations}
\end{equation}

\noindent This will be proved by induction using the maximum principle of $L_s$. For $i=1$, we recall that

\[L_s(\varphi_{+,s}) = \Delta \varphi_{+,s} - s^2k \varphi_{+,s} \leq c(s) -s^2 k \varphi_{+,s} -s^2 h e^{\varphi_{+,s}}=L_s(\varphi_{1,s}),\]

\noindent and therefore

\[L_s(\varphi_{1,s} - \varphi_{+,s}) \geq 0,\]

\noindent which implies $\varphi_{1,s} \leq \varphi_{+,s}$ by the Maximum Principle \ref{Maximum Principle}.
Suppose now that $\varphi_{i,s} \leq \varphi_{i-1,s}$. Since $k>-he^{\varphi_{+,s}}$ by its definition, one can
readily compute that

\begin{eqnarray}
L_s(\varphi_{i+1,s} - \varphi_{i,s}) &\geq& -s^2h e^{\varphi_{+,s}} \left[e^{\varphi_{i,s}-\varphi_{+,s}}-e^{\varphi_{i-1,s}-\varphi_{+,s}}-(\varphi_{i,s}-\varphi_{+,s})+(\varphi_{i-1,s}-\varphi_{+,s})\right] \nonumber \\
&=& -s^2he^{\varphi_{+,s}} \left[F(\varphi_{i,s}-\varphi_{+,s}) - F(\varphi_{i-1,s}-\varphi_{+,s})\right],
\label{iterative inequality for monotonicity}
\end{eqnarray}

\noindent where

\[F(x) = e^x -x.\]

\noindent $F(x)$ is a decreasing function when $x \leq 0$ since

\[F'(x) = e^x - 1 \leq 0 \;\; \forall x\leq 0.\]

\noindent Since $\varphi_{i,s}-\varphi_{+,s} \leq \varphi_{i-1,s}-\varphi_{+,s} \leq 0$ by inductive hypothesis,
we conclude that

\[\left[F(\varphi_{i,s}-\varphi_{+,s}) - F(\varphi_{i-1,s}-\varphi_{+,s})\right] \geq 0,\]

\noindent making the right hand side of \eqref{iterative inequality for monotonicity} nonnegative (recall that $-h>0$). This concludes the inductive step $\varphi_{i+1,s} \leq \varphi_{i,s} $ by the maximum principle of $L_s$. We finally show that

\[\varphi_{-,s} \leq \varphi_{i,s} \;\; \forall\; i,s.\]

\noindent This will again be shown by induction. To show that $\varphi_{-,s} \leq \varphi_{+,s}$ we suppose the
contrary, that

\[\inf_M \{\varphi_{+,s}(x)-\varphi_{-,s}(x)\} < 0.\]

\noindent Since $\varphi_{+,s} - \varphi_{-,s}$ is smooth and $M$ is compact, the infimum must be attained at some
point $x_0 \in M$. Therefore,

\[\Delta (\varphi_{+,s} - \varphi_{-,s}) (x_0) \leq -s^2h (e^{\varphi_{+,s}(x_0)} - e^{\varphi_{-,s}(x_0)}) < 0.\]

\noindent This is a contradiction since the Laplacian of a smooth function must be nonnegative at the minimum. We
conclude that $\varphi_{-,s} \leq \varphi_{+,s}$. Now suppose that $\varphi_{-,s} \leq \varphi_{i,s}$. Identical
computations as in \eqref{iterative inequality for monotonicity} yield

\[L_s(\varphi_{-,s} - \varphi_{i+1,s}) \geq -s^2he^{\varphi_{+,s}} \left[F(\varphi_{-,s} -\varphi_{+,s}) - F(\varphi_{i,s}-\varphi_{+,s})\right],\]

\noindent where $F(x)=e^x-x$ as above. Since $\varphi_{-,s} -\varphi_{+,s} \leq \varphi_{i,s}-\varphi_{+,s} \leq 0$
by inductive hypothesis, we again have $F(\varphi_{-,s} -\varphi_{+,s}) - F(\varphi_{i,s}-\varphi_{+,s}) \geq 0$
and therefore have established the inductive statement. The monotonicity relation
\eqref{monotinicity of iterative equations} is established.

Next, we wish to show that for each $s$, $\varphi_{i,s}$ uniformly converge to a smooth function $\varphi_s$.
This is a replica of argument from \cite{kw}. Recall inequality (3.12) from \cite{kw}, which is a consequence of Sobolev inequality (cf. (3.8) in \cite{kw}) and the fundamental elliptic regularity (cf. (3.9) in \cite{kw}). For $p>dim(M)$ and $u \in H^{2,p}$, we have

\begin{equation}
\Lnorm{u}{\infty} + \Lnorm{\nabla u}{\infty} \leq C_s \Lnorm{L_s u}{p}.
\label{elliptic regularity}
\end{equation}

\noindent Also recall that

\begin{equation}
\Lnorm{L_s(\varphi_{i,s})}{p} = \Lnorm{c(s) - s^2k\varphi_{i-1,s} -  s^2 h e^{\varphi_{i-1,s}}}{p}.
\label{Some equation}
\end{equation}

\noindent For a fixed $s$, \eqref{monotinicity of iterative equations} ensures that the right hand side of \eqref{Some equation} is uniformly bounded. Inequality \eqref{elliptic regularity} then implies that $\varphi_{i,s}$ and their first derivatives
are uniformly bounded in $L^\infty$. By the Theorem of Arzela-Ascoli, $\varphi_{i,s}$ possesses a subsequence
uniformly converging to a function $\varphi_s$ as $i \to \infty$. The monotonicity of $\varphi_{i,s}$ in $i$
implies that the subsequence is actually the entire sequence.

Moreover, the $L^p$ regularity shows that

\[\Wnorm{\varphi_{i+1,s}-\varphi_{j+1,s}}{2}{p} \leq C \left(\Lnorm{s^2h}{p} \Lnorm{e^{\varphi_{i,s}}-e^{\varphi_{j,s}}}{\infty}+\Lnorm{k}{p}\Lnorm{\varphi_{i,s}-\varphi_{j,s}}{\infty}\right).\]

\noindent For a fixed $s$, the sequence $\{\varphi_{i,s}\}_i$ converges in $L^\infty$, making the right hand side
of the inequality above Cauchy. The sequence  $\{\varphi_{i,s}\}_i$ is therefore strongly Cauchy in $H^{2,p}$, and
therefore strongly convergent. We have arrived at the conclusion that $\varphi_{i,s}$ converges to $\varphi_s$ in $H^{2,p}$, a
solution to the equation

\[\Delta \varphi_s = c(s) -s^2h e^{\varphi_s}.\]

\noindent Since $\varphi_s \in H^{2,p}$, $\Delta \varphi_s \in H^{2,p}$ as well since $h$ is smooth. Using Schauder's estimate (cf. section 3 of \cite{kw}), we conclude that $\varphi_s \in H^{4,p}$. Further bootstraping of the equation above implies that $\varphi_s \in H^{l,p}$ for all $l$. Since $M$ is compact, it follows that $\varphi_s \in H^{l,2}$ for all $l$, and is therefore smooth.

This completes the proof of Theorem \ref{Kazdan Warner Analysis} on the existence and uniqueness of the solutions to the equation.

\end{proof}

Consequentially, Theorem \ref{Identification of open subset} is proved.
\end{proof}

We now state the Main Theorem, on the bounded-ness and convergence of $\varphi_s$. Once again, this theorem is a
general analytic result. The functions and constants here need not be related to our initial geometric and
topological data. We nevertheless use the same notations for the convenience of application and comparison.

Before stating and proving the Main Theorem, we state the following elementary fact that follows easily from H\"older's inequality.

\begin{lemma} [Convergence of Powers]
Given two families of functions $\{f_s\}$ and $\{g_s\}$ on a bounded domain $U \subset \mathbb{R}^n$, such that $\{g_s\}$ are
uniformly bounded in $L^p$ for all $p$, and

\[\lim_{s \to \infty} \Lnorm{f_s-g_s}{p} = 0 \;\; \forall \; p,\]

\noindent we have

\[\lim_{s \to \infty} \Lnorm{f_s^N-g_s^N}{p} = 0 \;\; \forall \;p.\]

\noindent Here $N \in \mathbb{N}$ is arbitrary.
\label{convergence of powers}
\end{lemma}

\begin{theorem} [Main Theorem]
On a compact Riemannian manifold $M$ without boundary, let $c_1$ be any constant, $c_2$ any positive constant,
and $h$ any negative smooth function. Let $c(s) = c_1 -c_2 s^2$, for each $s$ large enough,
the unique solutions $\varphi_s \in C^\infty$ for the equations

\[
\Delta\varphi_{s}=c(s)-s^2he^{\varphi_s}.
\]

\noindent are uniformly bounded in $\W{l}{p}$ for all $l \in \mathbb{N}$ and $p \in [1,\infty]$. Moreover, in the limit
$s \to \infty$, $\varphi_{s}$ converges smoothly (i.e. in all $H^{l,p}$) to

\[\varphi_\infty = \log\left(\frac{c_2}{-h}\right),\]

\noindent the unique solution to

\[
he^{\varphi_{\infty}}+c_{2}=0.
\]

\label{Main Theorem}

\end{theorem}

\begin{proof}

We continue from the proof of the previous theorem. Recall the monotonicity and bounded-ness of $\varphi_{i,s}$:

\[\sub{s} \leq \varphi_{i,s} \leq \super{s}\]

\noindent for all $i,s$. Passing $i \to \infty$, we have

\begin{equation}
\sub{s} \leq \varphi_{s} \leq \super{s}
\label{uniform bound for varphi s}
\end{equation}

\noindent for all $s$. The functions $\sub{s}$ and $\super{s}$ are again uniformly bounded over $s$. In fact, one can observe that

\[\super{s} - \sub{s} = \log \left(\frac{K-c_1+c_2s^2}{-K-c_1+c_2s^2}\right) \to 0\]

\noindent in $L^\infty$ as $s \to \infty$. With the bounded-ness condition \eqref{uniform bound for varphi s}, we immediately conclude that

\[\varphi_\infty=\lim_{s\to \infty} \varphi_s = \lim_{s\to \infty}\super{s} = \lim_{s \to \infty}\sub{s} = \log\left(\frac{-c_2}{h}\right),\]

\noindent in $L^\infty$.

To show the convergence in general $\W{l}{p}$, we first consider a family of approximated solutions that converge smoothly to $\varphi_\infty$ as $s\to\infty$. Consider

\begin{equation}
v_s := \log \left(\frac{\Delta\left(-\log(-h)\right)-c(s)}{-s^2h}\right)
\label{defintion of approximated solutions}
\end{equation}

\noindent Since the function inside the logarithm converges smoothly to $\frac{-c_2}{h}$, it is clear that

\[v_s \to \varphi_\infty\]

\noindent smoothly as $s \to \infty$ . In fact, since all $v_s$ are uniformly bounded, Lemma \ref{convergence of powers} implies that

\begin{equation}
v_s^N \to \varphi_\infty^N,
\label{convergence of powers of functions}
\end{equation}

\noindent smoothly for all $N \in \mathbb{N}$ as $s\to \infty$.. These functions $v_s$ are approximated solutions to the PDE \eqref{Kazdan Warner Equation} in the following sense:

\[\Delta v_s = c(s) -s^2he^{v_s}+E_s,\]

\noindent where

\begin{equation}
E_s = \Delta \log \left(\frac{\Delta\left(-\log(-h)\right)-c(s)}{s^2}\right).
\label{definition of the error term}
\end{equation}

\noindent  Without the $h$ in the denominator, the function

\[\log \left(\frac{\Delta\left(-\log(-h)\right)-c(s)}{s^2}\right)\]

\noindent converge smoothly to a constant as $s\to \infty$ and therefore it is clear that $E_s \to 0$ smoothly as $s\to \infty$.

The convergence statement of the theorem then follows the lemma below:

\begin{lemma}
For all $l \in \mathbb{N}$, we have, with $v_s$ and $\varphi_s$ defined in this theorem, that

\[\lim_{s\to \infty} \Wnorm{\varphi_s - v_s}{l}{\infty} = 0.\]

\label{smooth convergence of function with approximations}
\end{lemma}

\begin{proof} ({\em of the Lemma})

We perform induction on $l$. The base case $l=0$ has been established, as both $v_s$ and $\varphi_s$ converge uniformly to $\varphi_\infty$ as $s\to \infty$. Before we establish the inductive step, we first make the following crucial claim.

\noindent \emph{Claim}:

\begin{equation}
\lim_{s\to \infty} \Lnorm{s^2\left(e^{\varphi_s}-e^{v_s}\right)}{\infty}=0
\label{convergence of varphi s and v s}
\end{equation}

To verify the claim, we start with the difference of the equations satisfied by $\varphi_s$ and $v_s$:

\begin{equation}
\Delta \left(\varphi_s-v_s\right) = -s^2he^{\varphi_s}+s^2h e^{v_s} - E_s
\label{diffrence of varphi s and v s}
\end{equation}

\noindent For each $s$, since the function $\varphi_s-v_s$ is smooth on the compact manifold $M$, there is a point $x_s \in M$ such that

\[\varphi_s(x_s)-v_s(x_s) = \sup_{x\in M}\{\varphi_s(x)-v_s(x)\}.\]

\noindent The Laplacian of $\varphi_s-v_s$ must be non-positive at $x_s$, and we have

\[0 \geq \Delta\left(\varphi_s - v_s\right)(x_s) = -s^2h(x_s)e^{\varphi_s(x_s)}+s^2h(x_s)e^{v_s(x_s)}-E_s(x_s).\]

\noindent It follows that, for all $x \in M$,

\begin{eqnarray}
E_s(x_s) &\geq& -s^2h(x_s)e^{v_s(x_s)}\left[e^{\varphi_s(x_s)-v_s(x_s)}-1\right] \nonumber \\
&\geq& -s^2h(x_s)e^{v_s(x_s)}\left[e^{\varphi_s(x)-v_s(x)}-1\right] \nonumber \\
&=& -s^2h(x_s)e^{v_s(x_s)}e^{-v_s(x)}\left[e^{\varphi_s(x)}-e^{v_s(x)}\right] \nonumber \\
\label{upper bound for the exponential difference construction}
\end{eqnarray}

\noindent The second inequality follows from the choice of $x_s$:

\[ \varphi_s(x_s)-v_s(x_s) \geq \varphi_s(x)-v_s(x) \;\; \forall x \in M.\]

\noindent Since the exponential function is monotonically increasing, and that $-s^2h(x_s)e^{v_s(x_s)}\geq0$, the inequality follows. We therefore arrive at the conclusion

\begin{equation}
s^2\left[e^{\varphi_s(x)}-e^{v_s(x)}\right] \leq E_s(x_s) \frac{e^{-v_s(x_s)}e^{v_s(x)}}{-h(x_s)}.
\label{upper bound for exponential difference}
\end{equation}

\noindent Since $v_s$ is uniformly convergent, thus bounded, and $h(x_s)\neq 0$, the fractional term is uniformly bounded. Since $E_s \to 0$ uniformly, the upper bound we have just obtained decays to 0 uniformly.

We need a lower bound that uniformly converge to 0. This is constructed using the same principle, except the special point $y_s \in M$ is chosen to be the point where the difference $\varphi_s - v_s$ achieves its infimum:

\[\varphi_s(y_s) - v_s(y_s) = \inf_{x\in M} \{\varphi_s(x)-v_s(x)\}.\]

\noindent The Laplacian of $\varphi_s - v_s$ now has to be non-negative at $y_s$, and we have identical chain of inequalities as in \eqref{upper bound for the exponential difference construction} in reverse order:

\begin{eqnarray}
E_s(y_s) &\leq& -s^2h(y_s)e^{v_s(y_s)}\left[e^{\varphi_s(y_s)-v_s(y_s)}-1\right] \nonumber \\
&\leq& -s^2h(y_s)e^{v_s(y_s)}\left[e^{\varphi_s(x)-v_s(x)}-1\right] \nonumber \\
&=& -s^2h(y_s)e^{v_s(y_s)}e^{-v_s(x)}\left[e^{\varphi_s(x)}-e^{v_s(x)}\right] \nonumber \\
\label{lower bound for the exponential difference construction}
\end{eqnarray}

This leads to the desired lower bound

\begin{equation}
s^2\left[e^{\varphi_s(x)}-e^{v_s(x)}\right] \geq E_s(y_s) \frac{e^{-v_s(y_s)}e^{v_s(x)}}{-h(y_s)},
\label{lower bound for exponential difference}
\end{equation}

\noindent which decays to 0 uniformly as $s \to \infty$. The decaying upper bound \eqref{upper bound for exponential difference} and lower bound \eqref{lower bound for exponential difference} verify the claim \eqref{convergence of varphi s and v s}.

Inductively, suppose that

\[\lim_{s\to \infty} \Wnorm{ \varphi_s -  v_s}{l}{\infty} = 0.\]

\noindent That is, for any multi-index $J$ such that $|J|\leq l$, we have

\[\lim_{s\to \infty} \Lnorm{\p^J \varphi_s -  \p^J v_s}{\infty} = 0.\]

\noindent We wish to establish the convergence to the order $l+1$. The proof is substantially identical to the one for Claim \eqref{convergence of varphi s and v s}, despite its involvement of rather tedious and lengthy bookkeeping of notations. Let $I$ be a multi-index of length $l+1$. We apply $\p^I$ to \eqref{diffrence of varphi s and v s}, with caution to the commutation relation between $\p^I$ and $\Delta$ stated in
\cite{V}, one computes

\begin{eqnarray}
& &\Delta\left(\p^{I} \varphi_s - \p^{I} v_s\right) \nonumber \\
&=&\sum_{j\in \M} \sum_{m^j(t)} \{\left[a_{m^j(t)} \left(\p^{I-j}h\right)\left(\p^{m_i} v_s\right)^{t_i} \right]s^2e^{v_s} - \left[a_{m^j(t)} \left(\p^{I-j}h\right)\left(\p^{m_i} \varphi_s\right)^{t_i} \right]s^2e^{\varphi_s} \}\nonumber \\
& &-h\left[\left(\p^{I}\varphi_s\right)s^2e^{\varphi_s}-\left(\p^{I}v_s\right)s^2e^{v_s}\right] +\sum_{j\in M^l}Q^j(Rm)\left(\p^j\varphi_s - \p^j v_s\right) \nonumber \\
& &- Q^I(Rm)\left(\p^I \varphi_s-\p^I v_s\right)-\p^{I}E_s. \nonumber \\
\label{l+1 derivative primitive}
\end{eqnarray}

\noindent Several notations above require explanations. These are algebraic expressions resulting from chain rules and product rules of differentiations, and the contributions of curvature tensors resulted form commuting $\p^I$ and $\Delta$. First,

\[M^l=\{r \in \mathbb{N}^n\;|\;|r|\leq l\},\]

\noindent so that $j\in\{I\}\cup M^l$ means exactly that $j=I$ or some multi-index of length no greater than $l$. Each $j \in \{I\}\cup M^l$ generates a collection of pairs of the form

\[m^j(t) := \{(m_1,\ldots,m_q),(t_1,\ldots,t_q)\;|\;m_i \in \mathbb{N}^n,t_i\in \mathbb{N}\}\]

\noindent such that $|m_i|\leq l$ and

\[m_1t_1+\ldots+m_qt_q=|j|.\]

\noindent $a_{m(t)}$'s are then the appropriate combinatorial constants in front of each function when differentiating the functions $e^{v_s}$ and $e^{\varphi_s}$ for $|j|$ times. For each $j$, $Q^j(Rm)$ is an algebraic combination of derivatives of the curvature tensors of $(M,g)$ up to $|j|^{th}$ order, and is therefore smooth and uniformly bounded. We may combine the $Q^j(Rm)$'s
into other terms in \eqref{l+1 derivative primitive} and rewrite it into:

\begin{eqnarray}
& &\Delta\left(\p^{I} \varphi_s - \p^{I} v_s\right) \nonumber \\
&=& -s^2he^{\varphi_s} \left[1-\frac{Q^I(Rm)}{s^2he^{\varphi_s}}\right]\left(\p^I \varphi_s - \p^I v_s\right) \nonumber \\
& &+\sum_{j\in \M}\left(A_{j,s}+B_{j,s}\right) \nonumber \\
& &+C_s-\p^I E_s, \nonumber \\
\label{l+1 derivative of the difference}
\end{eqnarray}

\noindent where

\begin{equation}
A_{j,s}=\sum_{m^j(t)\neq ((j),(1))} a_{m^j(t)}\left(\p^{I-j}h\right)\left[\left(\p^{m_i}v_s\right)^{t_i}
\left(s^2e^{v_s}-s^2e^{\varphi_s}\right)+\left(\left(\p^{m_i}v_s\right)^{t_i}-\left(\p^{m_i}\varphi_s\right)
^{t_i}\right)s^2e^{\varphi_s}\right],
\label{definition of A j s}
\end{equation}

\[B_{j,s}=
\begin{cases}

\left[a_{((j),(1))}\left(\p^{I-j}h\right)-\frac{Q^j(Rm)}{s^2e^{v_s}}
\right]s^2e^{v_s}\left(\p^j v_s\right)  \\
-\left[a_{((j),(1))}\left(\p^{I-j}h\right)-\frac{Q^j(Rm)}
{s^2e^{\varphi_s}}\right]s^2e^{\varphi_s}\left(\p^j\varphi_s\right);

&j \in M^l \\
0; & j=I
\label{definition of B j s}
\end{cases}
\]

\noindent and

\begin{equation}
C_s=-h\left(\p^I v_s\right)\left[s^2e^{\varphi_s}\right]\left(1-e^{v_s-\varphi_s}\right).
\label{definition of C s}
\end{equation}

\noindent One easily observes that for all $j$,

\begin{equation}
\lim_{s\to \infty} \Lnorm{\frac{A_{j,s}}{s^2}}{\infty}
= \lim_{s\to \infty} \Lnorm{\frac{B_{j,s}}{s^2}}{\infty}=\lim_{s\to \infty} \Lnorm{\frac{C_s}{s^2}}{\infty}=0.
\label{decay of A B C}
\end{equation}

\noindent The decays of $\frac{A_{j,s}}{s^2}$ and $\frac{C_s}{s^2}$ follow easily from inductive hypothesis
(all $j$ are of lengths no greater than $l$), Lemma \ref{convergence of powers}, Claim \eqref{convergence of varphi s and v s}, \eqref{convergence of powers of functions}, and the facts that $v_s$ are uniformly bounded in all Sobolev spaces and $\varphi_s$ is uniformly bounded in $L^\infty$.
These facts also imply the decay of $\frac{B_{j,s}}{s^2}$. Indeed, by Claim \eqref{convergence of varphi s and v s},
the smooth function $\rho_j(s)$ defined by

\[\rho_j(s):= \frac{Q^j(Rm)}{s^2e^{\varphi_s}}- \frac{Q^j(Rm)}{s^2e^{v_s}}\]

\noindent decays to 0 in $L^\infty$. One can then rewrite

\[B_{j,s} =
\begin{cases}
\left[a_{((j),(1))}\left(\p^{I-j}h\right)-\frac{Q^j(Rm)}{s^2e^{\varphi_s}}\right]
\left[\left(\p^j\varphi_s\right)\left(s^2e^{v_s}-s^2e^{\varphi_s}\right)+\left(\p^jv_s-\p^j\varphi_s\right)s^2e^{v_s}\right] \\
+\rho_j(s)s^2e^{v_s}\left(\p^jv_s\right); \\
& j \in M^l  \\
0; &j=I
\label{B j s rewrite}
\end{cases}\]

\noindent and the decay of $\frac{B_{j,s}}{s^2}$ in $L^\infty$ follows.

\noindent We are in the position to re-apply the maximum principle as in the base case $|I|=0$. Let $x_s \in M$ be the point so that

\[\p^{I}\varphi_s(x_s)-\p^{I}v_s(x_s) = \sup_{x\in M} \{\p^{I}\varphi_s(x)-\p^{I}v_s(x)\}.\]

\noindent Again, the Laplacian has to be non-positive at $x_s$, and we have, for all $x\in M$, that

\begin{eqnarray}
0&\geq&\Delta\left(\p^{I} \varphi_s - \p^{I} v_s\right)(x_s) \nonumber \\
&=& -s^2h(x_s)e^{\varphi_s(x_s)} \left[1-\frac{Q^I(Rm)}{s^2he^{\varphi_s}}\right](x_s)
\left(\p^I \varphi_s(x_s) - \p^I v_s(x_s)\right) \nonumber \\
& &+\sum_{j\in \M}\left(A_{j,s}(x_s)+B_{j,s}(x_s)\right) \nonumber \\
& &+C_s(x_s)-\p^I E_s(x_s), \nonumber \\
&\geq& -s^2h(x_s)e^{\varphi_s(x_s)} \left[1-\frac{Q^I(Rm)}{s^2he^{\varphi_s}}\right](x_s)
\left(\p^I \varphi_s(x) - \p^I v_s(x)\right) \nonumber \\
& &+\sum_{j\in \M}\left(A_{j,s}(x_s)+B_{j,s}(x_s)\right) \nonumber \\
& &+C_s(x_s)-\p^I E_s(x_s), \nonumber \\
\label{horrible inequality}
\end{eqnarray}

\noindent The two expressions before and after the second $\geq$ are identical except that we replace $x_s$ with
$x$ in the difference function $\p^I\varphi_s-\p^I v_s$ on the first line after the second $\geq$. For large enough $s$, we
have

\[1-\frac{Q^I(Rm)}{s^2he^{\varphi_s}} > 0\]

\noindent on $M$ and we may rearrange the \eqref{horrible inequality} without reversing the direction of inequalities:

\begin{eqnarray}
& &\p^I\varphi_s(x)-\p^Iv_s(x) \nonumber \\
&\leq& \frac{e^{-\varphi_s(x_s)}}{h(x_s) \left[1-\frac{Q^I(Rm)}{s^2he^{\varphi_s}}\right](x_s)}
\left(\sum_{j\in \M}\left[\frac{A_{j,s}(x_s)}{s^2}+\frac{B_{j,s}(x_s)}{s^2}\right]+\frac{C_s(x_s)}{s^2}-\frac{\p^I E_s(x_s)}{s^2}\right) \nonumber \\
\label{upper bound for l+1 derivative}
\end{eqnarray}

\noindent By \eqref{decay of A B C} and the fact that $E_s \to 0$ in all Sobolev spaces, the right hand side of this inequality
decays to 0 as $s \to \infty$.

The lower bound for $\p^{I}\varphi_s(x) - \p^{I}v_s(x)$ is obtained similarly. For each $s$, there is a special
point $y_s \in M$ such that

\[\p^{I}\varphi_s(y_s)-\p^{I}v_s(y_s) = \inf_{x\in M} \{\p^{I}\varphi_s(x)-\p^{I}v_s(x)\}.\]

\noindent The Laplacian of $\p^I \varphi_s - \p^I v_s$ has to be non-negative at $y_s$. Using identical arguments
as the ones for upper bound \eqref{horrible inequality} in reverse direction, we have, for all $x\in M$, that

\begin{eqnarray}
& &\p^I\varphi_s(x)-\p^Iv_s(x) \nonumber \\
&\geq& \frac{e^{-\varphi_s(y_s)}}{h(y_s) \left[1-\frac{Q^I(Rm)}{s^2he^{\varphi_s}}\right](y_s)}
\left(\sum_{j\in \M}\left[\frac{A_{j,s}(y_s)}{s^2}+\frac{B_{j,s}(y_s)}{s^2}\right]+\frac{C_s(y_s)}{s^2}-
\frac{\p^I E_s(y_s)}{s^2}\right) \nonumber \\
\label{lower bound for l+1 derivative}
\end{eqnarray}

\noindent The right hand side again decays to 0 uniformly as $s \to \infty$ with the same
arguments as in \eqref{upper bound for l+1 derivative}. Inequalities \eqref{upper bound for l+1 derivative}
 and \eqref{lower bound for l+1 derivative} establish the inductive step, and the lemma is therefore proved.

\end{proof}

With Lemma \ref{smooth convergence of function with approximations} established, the Main Theorem \ref{Main Theorem} follows trivially. Indeed, for all $l$,$p$, we have

\[\Wnorm{\varphi_s - \log\left(\frac{c_2}{-h}\right)}{l}{\infty} \leq \Wnorm{\varphi_s - v_s}{l}{\infty}+
\Wnorm{v_s - \log\left(\frac{c_2}{-h}\right)}{l}{\infty}\]

\noindent and the right hand side converge to 0 as $s \to \infty$. Theorem \ref{Main Theorem} then follows easily from the continuous embedding

\[\W{l}{\infty} \hookrightarrow \W{l}{p}\]

\noindent for any $l\in \mathbb{N}$ and $p \in [1,\infty]$.

\end{proof}

\section{Baptista's Conjecture}

We come back to the Riemann surface $M=\Sigma$. The results collected so far prove a conjecture posed by Baptista \cite{Ba}. It asserts that the natural $L^2$ metric on $\nu_{k,0}(s)$, when pulled back to $Hol_r(\Sigma,\mathbb{CP}^{k-1})$ via $\Phi_s$ described in Lemma \ref{BDW map}, evolves to a familiar one as $s \to \infty$. We prove this claim affirmatively. Throughout this section, we denote $z$ (and $\bar{z}$) as the local complex coordinate of $\Sigma$, $[z_0:\ldots:z_{k-1}]$ as the local homogeneous coordinates of $\pk$, and $(w_1,\ldots,w_m)$ as the local coordinate for $\Hol$, where $m=kr-(k-1)(b-1)$, as described in the remarks before Lemma \ref{BDW map}.

\subsection{The Evolution of $L^2$ Metrics on $\nu_{k,0}(s)$}

We start with the definition of natural $L^2$ metric on
$\mathcal{A}(H)\times \Omega^0(L)^{\oplus k}$, which is a special case of (4) in \cite{Ba}. At $(D_s,\phi_s) \in \mathcal{A}(H)\times \Omega^0(L)^{\oplus k}$, we define

\begin{equation}
g_{s}((\dot{A_{s}},\dot{\phi_{s}}),(\dot{A_{s}},\dot{\phi_{s}}))=\int_{\Sigma}\frac{1}{2s^{2}}\dot{A_s}\wedge
\bar{*}\dot{A_s}+<\dot{\phi_{s}},\dot{\phi_{s}}>_{H}vol_{\Sigma}.
\label{L 2 metric}
\end{equation}

\noindent Here, $(\dot{A}_s, \dot{\phi}_s)$ denotes a tangent vector in $T_{(D_s,\phi_s)}(\mathcal{A}(H)\times \Omega^0(L)^{\oplus k}) \simeq \Omega^1(\Sigma) \times \Omega^0(L)^{\oplus k}$. The identification is justified by the fact that $\Omega^0(L)$ is a vector space and $\mathcal{A}(H)$ is an affine space modeled on the vector space $\Omega^1(\Sigma)$, the space of complex valued one forms on $\Sigma$. (cf. Chapter V of \cite{K}). This identification also justifies the applications of Hodge star $\bar{*}$ and $<\cdot,\cdot>_H$ in the integrand of \eqref{L 2 metric}, since $(\dot{A_s},\dot{\phi_s})$ lies in essentially isomorphic spaces as $(D_s,\phi_s)$ does. By choosing tangent vectors orthogonal to $\mathcal{G}$-gauge transformations, \eqref{L 2 metric} descends to a well defined metric on the quotient space $\left(\mathcal{A}(H) \times \Omega^0(L)^{\oplus k}\right) / \mathcal{G}$ and restricts to the open subset $\nu_{k,0}(s)$.

The $L^2$ metric for $\Hol$ is also well known, with Fubini-Study metric endowed on $\mathbb{CP}^{k-1}$. Given $f \in Hol_r(\Sigma,\mathbb{CP}^{k-1})$, the tangent space of $\Hol$ at $f$ can be identified with the space of sections of the pullback bundle of $T \mathbb{CP}^{k-1}$ via $f$:

\[T_f\Hol \simeq \Gamma (f^* T\mathbb{CP}^{k-1}).\]

\noindent Given $u,v \in T_f \Hol$,  they can be viewed as a pullbacked sections on $\Sigma$, which can be pushed forward by $f$ to be tangent vectors on $\mathbb{CP}^{k-1}$, on which Fubini-Study metric can be applied. We define

\begin{equation}
\left<u,v\right>_{L^2}=\int_\Sigma \left<f_*u,f_*v\right>_{FS}vol_\Sigma.
\label{definition of L 2 metric on maps}
\end{equation}

\noindent Here, the $f_*$ denotes the pushforward of $f$.

Recall the diffeomorphic correspondence in Lemma \ref{BDW map}

\[\Phi_s: \Hol \to \nu_{k,0}(s).\]

\noindent We are interested in pulling back $g_s$ in \eqref{L 2 metric} to $Hol_r(\Sigma,\mathbb{CP}^{k-1})$ via $\Phi_s$, denoted by $g^*_s$, and comparing it with $<\cdot,\cdot>_{L^2}$ in \eqref{definition of L 2 metric on maps}. It was conjectured by Baptista that, roughly, $g_s$ approaches a constant multiple of $<\cdot,\cdot>_{L^2}$ as $s \to \infty$.

We carefully list the required data to proceed our analysis. Start with a holomorphic map $\tilde{\phi}: \Sigma \to \mathbb{CP}^{k-1}$. Equip $\mathbb{CP}^{k-1}$ with the standard Fubini-Study metric $g_{FS}$. On the coordinate chart $U_i \subset \pk$ where $z_i \neq 0$, the expression of K\"ahler form of $g_{FS}$ is well known:

\begin{equation}
\omega_{FS} = \frac{\sqrt{-1}}{2\pi} \p \bar{\p} \log \left(\sum_{l=0}^n \left|\frac{z_l}{z_i}\right|^2\right).
\label{Kahler Form of Fubini Study}
\end{equation}

\noindent This form is also known to be globally defined. There is then a natural Hermitian metric on $\Oone$ whose curvature form is a multiple of $\omega_{FS}$. Explicitly, the metric is given locally at $[z_0:\ldots:z_{k-1}] \in \mathbb{CP}^{k-1}$ by

\[H_{FS}(\cdot) := \frac{1}{\sum_{i=1}^k |z_i|^2} |\cdot|^2,\]

\noindent where $|\cdot|$ is the standard Euclidean flat metric in the local trivialization of $\Oone$ over $U_i$. $H_{FS}$ carries the feature that its curvature form $F_{FS}$ satisfies

\[\frac{\sqrt{-1}}{2\pi}  F_{FS} = \omega_{FS}.\]

\noindent Therefore,

\[\sqrt{-1} F_{FS} = \frac{1}{2\pi} (\omega_{FS},\omega_{FS})_{\omega_{FS}} = \frac{1}{2\pi}.\]

\noindent See, for example, section 1.2 of \cite{GH} for more details. $\omega_{FS}$ is the generator for $H^2(\pk,\mathbb{Z})$, that is, $\int_{\pk} [\omega_{FS}]^{k-1} =1$. In \cite{Ba}, the author used the convention for the K\"ahler form $\omega_{\pk} = \pi \omega_{FS}$, and referred to the normalized form $\omega_{FS}$ as $\omega_{\text{norm }FS}$.

Recall the pullback construction of the line bundle $L$, sections $\phi$, and background Hermitian metric arisen from $\tilde{\phi}$, as in Lemma \ref{BDW map}:

\[
\begin{diagram}
\node{(L,H)}\arrow{s} \node{(\Oone,H_{FS})}  \\
\node{\Sigma} \arrow{e,t}{\tilde{\phi}} \node{\mathbb{CP}^{k-1}} \arrow{n,r}{s_1,\ldots,s_k}
\end{diagram}
\]

\noindent where $L:=\tilde{\phi}^*\Oone$ and $H:=\tilde{\phi}^*H_{FS}$.

The global hyperplane sections $s_1,\ldots,s_k$ on $\Oone$ are pulled back to $L$:

\[\phi := (\phi_i := \tilde{\phi}^*s_i)_i,\]

\noindent and $\tilde{\phi}$ also defines a holomorphic structure $\bar{\p}_L$ by pulling back the standard complex structure $\bar{\p}_{\mathbb{CP}^{k-1}}$ on $\Oone$. By the definition of $H_{FS}$ on $\Oone$, it is automatic that

\begin{equation}
\sum_{i=1}^k |\phi_i|^2_H = 1.
\label{constant norm}
\end{equation}

We describe the variations of holomorphic maps and their corresponding pushforwards on $\nu_{k,0}(s)$. Given $\dot{\tilde{\phi}} \in T_{\tilde{\phi}}Hol_r(\Sigma,\mathbb{CP}^{k-1})\simeq \Gamma (\tilde{\phi}^* T\mathbb{CP}^{k-1})$, we construct a smoothly varying curve $\tilde{\phi}(t)$ in $Hol_r (\Sigma,\mathbb{CP}^{k-1})$ so that $\tilde{\phi}(0)=\tilde{\phi}$ and $\frac{\p}{\p t}|_{t=0}\tilde\phi (t)=\dot{\tilde{\phi}}$. $\tilde{\phi}(t)$ is locally expressed on $U_i$ as

\begin{equation}
\tilde{\phi}(t)=\left[\tilde{\phi_1}(t),\ldots,\tilde{\phi_k}(t)\right].
\label{curve of holomorphic maps}
\end{equation}

The corresponding family of sections in $\nu_{k,0}(s)$ are then defined by pulling back the global sections $s_1,\ldots,s_k$ via $\tilde{\phi}(t)$:

\begin{equation}
\phi_t=\left[\phi_{1,t},\ldots,\phi_{k,t}\right] \in \Omega^0(L) \times \ldots \times \Omega^0(L),
\label{curve of sections}
\end{equation}

\noindent where

\[\phi_{i,t}  := \left(\tilde{\phi}(t)\right)^*(s_i).\]

\noindent Taking $t$-derivatives of $\tilde{\phi}(t)$'s at $t=0$, we obtain

\begin{equation}
\pderi \tilde{\phi}(t) = \dot{\tilde{\phi}} = \left(\pderi \tilde{\phi}_1(t),\ldots,\pderi \tilde{\phi}_n(t)\right) \in \tilde{\phi}^*\left(T\pk\right).
\label{t derivative of maps}
\end{equation}

\noindent To identify the corresponding infinitesimals on $T_{[D_s,\phi_s]}\nu_{k,0}(s)$, we recall the classical short exact Euler sequence of bundles over $\pk$, summarized from section 3.3 of \cite{GH}:

\begin{equation}
\begin{tikzpicture}[start chain] {
    \node[on chain] {$0$};
    \node[on chain] {$\mathcal{O}_{\pk}$};
    \node[on chain, join={node[above] {$ \iota $}}] {$\Oone^{\oplus k}$};
    \node[on chain, join={node[above] {$\mathcal{E}$}}] {$T\pk$};
    \node[on chain] {$0$};, }
\end{tikzpicture}
\label{Euler sequence}
\end{equation}

\noindent where $\mathcal{O}_{\pk}$ is the trivial line bundle. The map $\iota$ is obtained by twisting the natural inclusion $\mathcal{O}_{\pk}(-1) \subset \mathcal{O}_{\pk}^{\oplus k}$ with $\mathcal{O}_{\pk}(1)$. For the map $\mathcal{E}$, we take $\sigma:=(\sigma_1,\ldots,\sigma_k)\in  \Gamma \left(\Oone^{\oplus k}\right)$, $z=[z_0,\ldots,z_{k-1}] \in U_i \subset \pk$, and $Z:=(Z_0,\ldots Z_k) \in \mathbb{C}^{k} - \{0\}$ so that $\pi(Z)=z$. Here, $\pi$ is the natural projection from $\mathbb{C}^{k} - \{0\}$ to $\pk$. We then define

\begin{equation}
\mathcal{E} (\sigma)|_z = \pi_* \left(\sum_i \sigma_i(z) \frac{\p}{\p Z_i}\right),
\label{surjection of Euler sequence}
\end{equation}

\noindent which is a linear map with kernel

\[\ker \mathcal{E}=\{a(Z_0,\ldots,Z_{k-1})\;\;|\;\;a\in \mathbb{C}\}.\]

\noindent Indeed, the tangent space $T_z \pk$ is spanned by $\{\pi_* \frac{\p}{\p Z_i}\}_{i=0}^{k-1}$ subject to the relation

\[\sum_i Z_i \frac{\p}{\p Z_i} =0.\]

\noindent Setting $a=0$, a section of $T\pk$ is then uniquely associated with a $k$-tuple of global sections of $\mathcal{O}_{\pk}(1)$.

Pulling back the Euler sequence \eqref{Euler sequence} by $\tilde{\phi}$, we obtain a short exact sequence of bundles over $\Sigma$:

 \begin{equation}
\begin{tikzpicture}[start chain] {
    \node[on chain] {$0$};
    \node[on chain] {$\tilde{\phi}^*\mathcal{O}_{\pk}$};
    \node[on chain, join={node[above] {$ \tilde{\phi}^*\iota $}}] {$\tilde{\phi}^*\Oone^{\oplus k}$};
    \node[on chain, join={node[above] {$\tilde{\phi}^*\mathcal{E}$}}] {$\tilde{\phi}^*T\pk$};
    \node[on chain] {$0$};. }
\end{tikzpicture}
\label{Euler sequence pull back}
\end{equation}

\noindent In particular, for $\dot{\tilde{\phi}} \in \Gamma\left(\tilde{\phi}^*T\pk\right)$, the correspondence just discussed associates a unique $k$-tuple of global sections on $L$, denoted by

\[\dot{\phi} := \left(\dot{\phi}_1,\ldots,\dot{\phi}_k\right) \in \Omega^0(L)^{\oplus k}=L^{\oplus k}.\]

 \noindent The family of holomorphic maps $\tilde{\phi}(t)$ also defines a family of line bundles over $\Sigma$:

\[L_t:=\tilde{\phi}(t)^* \Oone.\]

\noindent All bundles are of degree $r$ and therefore isomorphic as complex line bundles. However, each of them is equipped with its own pullback holomorphic structure:

\[\bar{\p}_{L_t} := \tilde{\phi}(t)^* \left(\bar{\p}_{\mathbb{CP}^{k-1}}\right).\]

\noindent Clearly

\[\bar{\p}_L =\bar{\p}_{L_0}.\]

\noindent Each $L_t$ is equipped with a background metric

\[H_t:=\tilde{\phi}(t)^* H_{FS}\]

\noindent and we denote $H:=H_0$.

To analyze $g_s^*$, we need to compute the  pushforward of $\dot{\tilde{\phi}}$ under $\Phi_s$, denoted by $\left(\dot{A_s},\dot{\phi_s}\right)$ as in \eqref{L 2 metric}. For each $t$, our constructions above clearly imply

\[\bar{\p}_{L_t} \phi_{i,t} = 0 \;\; \forall t,i.\]

\noindent By Theorem \ref{Identification of open subset}, we can then find a unique gauge $e^{2u_{s,t}} \in \mathcal{G}_{\mathbb{C}}$ so that

\[\left[D(e^{u_{s,t}}{ }^* \bar{\p}_{L_t}),e^{u_{s,t}}\phi_t \right]\in \nu_{k,0}(s),\]

\noindent where $D(e^{u_{s,t}}{ }^* \bar{\p}_{L_t})$ is the unique $H$-unitary connection compatible with the holomorphic structure $\bar{\p}_{L_t}$. (Readers may review Lemma \ref{identification of moduli space} for the detailed descriptions.) The map $\Phi_s$ is now explicitly written for each $t$:

\[\Phi_s(\tilde{\phi}(t)) = \left[D(e^{u_{s,t}}{ }^* \bar{\p}_{L_t}),e^{u_{s,t}}\phi_t\right].\]

\noindent Recall the gauge action on holomorphic structures:

\[e^{u_{s,t}}{ }^* \bar{\p}_{L_t} = e^{u_{s,t}} \left(\bar{\p}_L e^{-u_{s,t}}\right) = \bar{\p}_{L_t} -\left(\frac{\p u_{s,t}}{\p \bar{z}}\right)d\bar{z},\]

\noindent we have

\begin{equation}
\Phi_s(\tilde{\phi}(t)) = \left[D\left(e^{u_{s,t}} (\bar{\p}_L e^{-u_{s,t}})\right),e^{u_{s,t}}\phi_t\right],
\label{definition of Phi s explicitly}
\end{equation}

\noindent where $D\left(e^{u_{s,t}} (\bar{\p}_L e^{-u_{s,t}})\right)$ is the $H_t$-unitary connection with respect to the holomorphic structure

\[e^{u_{s,t}} \left(\bar{\p}_L e^{-u_{s,t}}\right).\]

At $t=0$, the pushforward of $\dot{\tilde{\phi}}$ can now be readily computed:

\begin{equation}
\dot{\phi}_s = e^{u_s} \dot{\phi}+e^{u_s}\dot{u_s}\phi,
\label{explicit variation of sections}
\end{equation}

\noindent where

\[\dot{u_s} := \frac{\p}{\p t}|_{t=0} u_{s,t}.\]

\noindent \eqref{explicit variation of sections} makes sense as $\phi$ and $\dot{\phi}$ reside in the same space.

$\dot{A_s}$ needs to be computed with caution. Let $\gamma \in \Omega^0(U,L)$ be a local holomorphic frame for $L$ over an open chart $U$, with respect to the holomorphic structure $\bar{\p}_L$. The background Hermitian metric is locally given by a smooth function $H_t$ in this setting. Altering the holomorphic structure, we observe that the section $e^{u_{s,t}} \gamma$ is the local holomorphic frame with respect to the holomorphic structure $e^{u_{s,t}} (\bar{\p}_L e^{-u_{s,t}})$. With respect to this frame, the same background Hermitian metric now has local coordinate description by the smooth function

\[H_t^\prime = H_t e^{2u_{s,t}}.\]

\noindent We then compute the connection form $A_{s,t}$ of $D(e^{u_{s,t}} (\bar{\p}_L e^{-u_{s,t}}))$ using the standard formula of the local expression of $H_t^\prime$-unitary connection forms over $U$ (cf. I.(4.11) in \cite{K}):

\begin{eqnarray}
A_{s,t} &=& (H_t^\prime)^{-1} \p (H_t^\prime) \nonumber \\
       &=& \frac{\left(\frac{\p H_t}{\p z} + 2H_t\frac{\p u_{s,t}}{\p z}\right)}{H_t} dz \nonumber \\
       &=& \left[\frac{\p}{\p z} \left(\log H_t\right) + 2 \frac{\p u_{s,t}}{\p z}\right]dz. \nonumber \\
       \label{connection form explicit descriptions}
\end{eqnarray}

\noindent We differentiate $A_{s,t}$ with respect to $t$ and evaluating it at $t=0$ to obtain $\dot{A_s}$:

\begin{equation}
\dot{A_s} := \frac{\p}{\p t}|_{t=0}A_{s,t} = \frac{\p}{\p z}\left(\frac{\dot{H}}{H}\right)+2 \frac{\p \dot{u_s}}{\p z} dz,.
\label{pushforward of connection forms}
\end{equation}

\noindent where

\[\dot{H} := \frac{\p}{\p t}|_{t=0} H_t.\]

We have now identified the pushforward of the $\dot{\tilde{\phi}}$ under $\Phi_s$:

\begin{eqnarray}
\Phi_{s,*}\left(\dot{\tilde{\phi}}\right) &=& \left(\dot{A_s},\dot{\phi_s}\right) \nonumber \\
&=& \left(\frac{\p}{\p z}\left(\frac{\dot{H}}{H}\right)+2 \frac{\p \dot{u_s}}{\p z} dz,
e^{u_s}\dot{\phi}+e^{u_s}\dot{u_s}\phi\right) \in T_{[D_s,\phi_s]}\nu_{k,0}(s) \nonumber \\
\label{pushforward definition}
\end{eqnarray}

\noindent By the definition of pullback metric, we then have

\begin{eqnarray}
g_s^*\left(\dot{\tilde{\phi}},\dot{\tilde{\phi}}\right) &:=& g_s\left(\Phi_{s,*}\left(\dot{\tilde{\phi}}\right),\Phi_{s,*}\left(\dot{\tilde{\phi}}\right)\right) \nonumber \\
&=& \int_\Sigma \left(\frac{\left|\frac{\p}{\p z}\left(\frac{\dot{H}}{H}\right)+2\frac{\p \dot{u_s}}{\p z}\right|^2}{2s^2} + \left<\dot{\phi},\dot{\phi}\right>_H e^{2u_s} + \left(e^{u_s}\dot{u_s}\right)^2\right) vol_\Sigma, \nonumber \\
\label{L 2 metric pull back}
\end{eqnarray}

This quantity is a real number since $z$, the coordinate of $\Sigma$, is eliminated after integration over $\Sigma$. The second equality above relies the relation \eqref{constant norm}, which implies $\left<\phi,\dot{\phi}\right>_H=0$ and $\left<\phi,\phi\right>_H=1$. One should expect the first and third terms in \eqref{L 2 metric pull back} to vanish as $s \to \infty$, and the second term to approach a multiple of square norm of $\dot{\phi}$. Namely, we expect \eqref{L 2 metric pull back} to approach the (multiple of) $<\cdot,\cdot>_{L^2}$ defined in \eqref{definition of L 2 metric on maps}. This is precisely the statement in the Baptista's Conjecture in \cite{Ba}.

\begin{conjecture}[Baptista's Conjecture]
On $Hol_r(\Sigma,\mathbb{CP}^{k-1})\backsimeq \nu_{k,0}(s) $, $g^*_s$ defined in \eqref{L 2 metric pull back} converges smoothly, as $s \to \infty$, to a multiple of the ordinary $L^2$ metric $<\cdot,\cdot>_{L^2}$ defined in \eqref{definition of L 2 metric on maps} on $Hol_r(\Sigma,\mathbb{CP}^{k-1})$.
\end{conjecture}

To achieve higher mathematical precision, we state the following notion of convergence.

\begin{definition} [Cheeger-Gromov Convergence] For all $l \in \mathbb{N}$ and $p \geq 1$, a family of $n$-dimensional Riemannian manifolds
$(M_s,g_s)$ is said to converge to a fixed Riemannian manifold $(M,g)$ in $\W{l}{p}$, in the sense of
Cheeger-Gromov, if there is a covering chart $\{U_k, (x_i^k)\}$ on $M$ and a family of
diffeomorphisms $F_s:M \to M_s$, such that

\begin{equation}
 \Wnormloc{F_s^*(g_s)(\frac{\partial}{\partial x_i},\frac{\partial}{\partial x_j}) - g(\frac{\partial}{\partial x_i},\frac{\partial}{\partial x_j})}{l}{p}{U_k}\to 0.
 \label{Cheeger Gromov convergence}
 \end{equation}

\noindent as $s \to \infty$, for all $k$ and $i,j \in \{1,\ldots,n\}$.
\end{definition}

\noindent We state Baptista's Conjecture in this level of mathematical rigor:

\begin{proposition}
[Precise Baptista's Conjecture]
Equipping $\mathbb{CP}^{k-1}$ with the Fubini-Study metric, the sequence of metrics $g_s$ on $\nu_{k,0}(s)$ given by \eqref{L 2 metric} converges in all $H^{l,p}$ (and so smoothly), in the sense of Cheeger-Gromov, to $\frac{1}{2}$ times the $L^2$ metric $<\cdot,\cdot>_{L^2}$ on $Hol_r(\Sigma,\mathbb{CP}^{k-1})$ given by \eqref{definition of L 2 metric on maps}. The family of diffeomorphisms are precisely $\Phi_s$, as constructed in Lemma \ref{BDW map}.
\label{Precise Baptista's Conjecture}
\end{proposition}

\begin{proof}

Throughout the proof, we use the following abbreviations for the initial value and variation of a family of functions $f_t$ with parameter $t$:

\[f := f_0,\]

\noindent and

\[\dot{f} := \pderi f_t.\]

We first recall that for each $t$, the $k$-sections $\phi_t$ give rise to the function

\begin{equation}
h_t = -e^{2\psi_t}\sum_{i=1}^k  |\phi_{i,t}|_{H_t}^2
\label{the function h for Baptista}
\end{equation}

\noindent as in \eqref{equaiton for psi} and \eqref{Kazdan Warner Equation} of section 3, where

\[\Delta \psi_t = \sqrt{-1} \Lambda F_{H_t} -c_1.\]

\noindent However, since $\phi_t$ and $H_t$ are pullbacked from the sections $s_1,\ldots,s_k$ on $\Oone$ with constant $H_{FS}$ norm of 1, it is clear that $\sum_{i=1}^k  |\phi_{i,t}|_{H_t}^2=1\;\;\forall t$, and

\[h_t=-e^{2\psi_t}.\]

For each $t$, recall the relation of $u_{s,t}$ and $\varphi_{s,t}$:

\[\varphi_{s,t}=2(u_{s,t}-\psi_t).\]

\noindent It follows that $e^{2u_{s,t}}=-h_te^{\varphi_{s,t}}$ and

\[\frac{\p \dot{u_s}}{\p z}=\frac{1}{2} \left(\frac{\p \dot{\varphi_s}}{\p z}+2\frac{\p \dot{\psi}}{\p z}\right).\]

\noindent The pullback metric $g_s^*$ \eqref{L 2 metric pull back} can be rewritten as

\begin{equation}
g_s^*\left(\dot{\tilde{\phi}},\dot{\tilde{\phi}}\right)=\int_\Sigma \left(\frac{\left|\frac{\p}{\p z} \left(\frac{\dot{H}}{H}\right)+\frac{\p \dot{\varphi_s}}{\p z}+2\frac{\p \dot{\psi}}{\p z}\right|^2}{2s^2}+\left<\dot{\phi},\dot{\phi}\right>_H\left(-he^{\varphi_{s}}\right) -\frac{1}{2}\dot{\left(he^{2\varphi_s}\right)}\left(\dot{u_s}\right)\right)vol_\Sigma
\label{L 2 metric pull back alternative}
\end{equation}

\noindent The dot over $-he^{2\varphi_s}$ above is applied to the entire product. It is also evident from our constructions that

\[\left<\dot{\phi},\dot{\phi}\right>_H=\left<\dot{\tilde{\phi}},\dot{\tilde{\phi}}\right>_{H_{FS}}.\]

We now allow $\tilde{\phi}$ to vary freely on $\Hol$. Each $\tilde{\phi} \in \Hol$ determines corresponding Hermitian structures and functions $H_{\tilde{\phi}}$, $h_{\tilde{\phi}}$, $\psi_{\tilde{\phi}}$, $u_{s,\tilde{\phi}}$, and $\varphi_{\tilde{\phi}}$ on $\Sigma$, as in the constructions in section 3. The subscripts did not appear there since we fixed a holomorphic function to begin the entire argument. To emphasize the variation on $\Hol$ in the present situation, we amend the notations of the functions discussed in section 3:

\[\tilde{H}, \tilde{h}, \tilde{\psi}, \tilde{u_s}, \tilde{\varphi_s} : \Hol \times \Sigma \to \mathbb{R},\]

\noindent so that $\tilde{H}(\tilde{\phi},z)=H_{\tilde{\phi}}(z)$ and similarly for other functions. These functions are all smooth, as their dependencies on holomorphic maps are smooth.

Before establishing the convergence, we note that it is sufficient to prove the convergence of $g_s^*$ in a coordinate neighborhood $\mathcal{U}$ of $\tilde{\phi}$, as Cheeger-Gromov convergence is a local statement. Moreover, using polarizing identity of the Hermitian structure, it is sufficient to establish the convergence \eqref{Cheeger Gromov convergence} for some $i=j$. Fix a precompact coordinate patch $\mathcal{U} \subset \Hol$ with coordinates $(w_1,\dots,w_m)$ centered at $\tilde{\phi}$. We remind the readers that the coordinate description of $\Hol$ is given in the remark immediately before Lemma \ref{BDW map}. Let

\[\left(\xi_1,\ldots,\xi_m\right)\]

\noindent be the coordinate local frame of $T\Hol$ over $\mathcal{U}$ so that for all $f \in C^\infty(\Hol)$ and all $\tilde{\eta} \in \mathcal{U}$,

\[\xi_i (\tilde{\eta}) (f) = \frac{\p}{\p w_i}|_{\tilde{\eta}} f.\]

\noindent Setting $\dot{\tilde{\phi}} = \xi_i$ in \eqref{L 2 metric pull back alternative} then defines a real smooth function on $\mathcal{U}$. Precisely, at $\tilde{\eta} \in \mathcal{U}$, we define

\begin{eqnarray}
F_s^i(\tilde{\eta}) &:=& g_s^*\left(\xi_i(\tilde{\eta}),\xi_i(\tilde{\eta}) \right) \nonumber \\
&=& \int_\Sigma \left[\frac{\left|\frac{\p}{\p z} \left(\wpderi \log \tilde{H}\right)+\frac{\p}{\p z}  \left(\wpderi \tilde{\varphi_s}\right) + 2\frac{\p}{\p z}\left(\wpderi \tilde{\psi}\right)\right|^2}{2s^2}\right] vol_\Sigma \nonumber \\
& &+\int_\Sigma \left[\left<\xi_i(\tilde{\eta}),\xi_i(\tilde{\eta}) \right>_{H_{FS}}\left[-\tilde{h}(\tilde{\eta},z)\right]e^{\tilde{\phi_s}(\tilde{\eta},z)}\right]\nonumber\\
& &-\frac{1}{2}\int_\Sigma\left(\wpderi\tilde{h}e^{2\tilde{\varphi_s}}\right)\left(\wpderi\tilde{u_s}\right) vol_\Sigma. \label{L 2 metric pull back free}
\end{eqnarray}

\noindent Once again, the $z$ variable is integrated out on the right hand side and $F_s^i$ is solely a function on $\mathcal{U}$.

The derivatives of $F_s^i$ can be computed accordingly. For a multi-index $R \in \mathbb{N}^m$, we may compute

\[\p^R F_s^i.\]

\noindent Here, as usual, the multi-index convention is adopted. For $R=(r_1,\ldots,r_m)$,

\[\p^R :=\frac{\p^{r_1}\cdots\p^{r_m}}{\p w_1^{r_1} \cdots \p w_m^{r_m}}.\]

\noindent In this section, we reserve this notation for differentiations on coordinates of $\mathcal{U}$ only. Since all functions in the integrand of \eqref{L 2 metric pull back free} are smooth, we may interchange $\p^R$ with the integration:

\begin{eqnarray}
& &\p^R F_s^i (\tilde{\eta}) \nonumber \\
&=&\int_\Sigma  \frac{\pRi \left| \frac{\p}{\p z} \left(\wpder \log \tilde{H} \right) +  \frac{\p}{\p z}\wpder \tilde{\varphi_{s}}-2 \frac{\p}{\p z}\left(\wpder \tilde{\psi}\right)\right|^2}{2s^2} vol_\Sigma \nonumber \\
& &+\int_\Sigma \left[\left[\p^R\left<\xi_i,\xi_i\right>_{H_{FS}}\right]\left(-\tilde{h}e^{2\tilde{\varphi_{s}}}\right)\right]\vline_{(\tilde{\eta},z)}  vol_\Sigma \nonumber \\
& &-\int_\Sigma \sum_{\alpha \in \R}\left[A_\alpha\p^\alpha\left(\tilde{h}e^{\tilde{\varphi_{s}}}\right) B_{R-\alpha}\right]\vline_{(\tilde{\eta},z)} vol_\Sigma. \nonumber \\
& &-\frac{1}{2}\int_\Sigma \pRi \left[ \left(\wpder\tilde{h}e^{2\tilde{\varphi_s}}\right)\left(\wpder\tilde{u_s}\right)\right]
\label{derivative of L 2 metric}
\end{eqnarray}

\noindent Here, again, $M^R$ is the set of all multi-indices with lengths less than $|R|$, as defined in the proof of Lemma \ref{smooth convergence of function with approximations}. $B_{R-\alpha}$ are smooth functions defined by

\[B_{R-\alpha}=\p^{R-\alpha}\left<\xi_i,\xi_i\right>_{H_{FS}},\]

\noindent which are independent of $s$. $A_\alpha$'s are constants. From the expression of \eqref{derivative of L 2 metric}, the conclusion of Proposition \ref{Precise Baptista's Conjecture} then holds true on $\mathcal{U}$ once we verify the following three conditions for all multi-index $R$ and all $\tilde{\eta}\in \mathcal{U}$:

\begin{equation}
\lim_{s\to\infty}\Lnormloc{\frac{\pRi \left| \frac{\p}{\p z} \left(\wpder \log \tilde{H} \right) +  \frac{\p}{\p z}\wpder \tilde{\varphi_{s}}-2 \frac{\p}{\p z}\left(\wpder \tilde{\psi}\right)\right|^2}{2s^2}}{\infty}{\Sigma}=0;
\label{convergence 1}
\end{equation}

\begin{equation}
\lim_{s\to\infty}\Lnormloc{\p^\alpha\left(\tilde{h}e^{2\tilde{\varphi_{s}}}\right)}{\infty}{\Sigma}=0 \;\;\forall \alpha \text{ such that }1\leq |\alpha| \leq |R|;
\label{convergence 2}
\end{equation}

\noindent and

\begin{equation}
-\tilde{h}e^{\tilde{\varphi_s}}|_{(\tilde{\eta},z)} \to \frac{1}{2} \text{ in } L^\infty(\Sigma) \text{ as } s \to \infty.
\label{convergence 6}
\end{equation}

\noindent Here, $\|\cdot\|_{L^\infty(\Sigma)}$ denotes the $L^\infty$ norm of the space of uniformly bounded functions on $\Sigma$, where $\Sigma$ is amended to emphasize the fact that after evaluating the three expressions above at a particular point $\tilde{\eta} \in \mathcal{U}$, they are functions solely on $\Sigma$. \eqref{convergence 6} follows directly from the Main Theorem \ref{Main Theorem}. To verify the other two claims,
we similarly define the approximated solutions $\tilde{v}_s$ and error $\tilde{E}_s$ on $\mathcal{U}\times \Sigma$ as in the proof of the Main Theorem \ref{Main Theorem}:

\begin{equation}
\tilde{v}_s := \log \left(\frac{\Delta_\Sigma\left(-\log(-\tilde{h})\right)-c(s)}{-s^2\tilde{h}}\right)
\label{definition of approximated solutions t}
\end{equation}

\noindent with error

\begin{equation}
\tilde{E}_s := \Delta_\Sigma \log \left(\frac{\Delta_\Sigma \left(-\log(-\tilde{h})\right)-c(s)}{s^2}\right)
\label{definition of the error term t}
\end{equation}

\noindent so that

\[\Delta_\Sigma \tilde{v}_{s} + s^2\tilde{h}e^{\tilde{v}_{s}} -c(s) = \tilde{E}_{s}.\]

\noindent Here, $\Delta_\Sigma$ denotes the Laplacian with respect to coordinates of $\Sigma$ only and $c(s)=2c_1-\frac{1}{2}s^2$ as in the beginning of section 3. One can readily see that for all $R\in \mathbb{N}^m$ and $\tilde{\eta}\in \mathcal{U}$,

\begin{equation}
\Wnormloc{\pRi \tilde{v}_{s}}{l}{\infty}{\Sigma}\leq C_R<\infty\;\;\forall \;s,
\label{convergence 3}
\end{equation}

\begin{equation}
\lim_{s\to\infty}\Lnormloc{\p^\alpha|_{(\tilde{\eta},z)} \left(\tilde{h}e^{\tilde{v}_{s}}\right)}{\infty}{\Sigma}=0 \;\;\forall \alpha \text{ such that }1\leq|\alpha|\leq |R|,
\label{convergence 4}
\end{equation}

\begin{equation}
-\tilde{h}e^{\tilde{v_s}}|_{(\tilde{\eta},z)} \to \half \text{ in } L^\infty(\Sigma) \text{ as } s\to \infty,
\label{convergence 5}
\end{equation}

\noindent and

\begin{equation}
\lim_{s\to\infty}\Wnormloc{\pRi \tilde{E}_{s}}{l}{\infty}{\Sigma}=0,
\label{convergence 7}
\end{equation}

\noindent $\forall l \in \mathbb{N}$, where we again use the amended notation $H^{l,\infty}(\Sigma)$ to denote the space of functions on $\Sigma$ with uniformly bounded derivatives up to $l^{th}$ order. All claims follow from direct computations of derivatives. To bound $\pRi \tilde{v_s}$, we observe that the argument of $\log$ in \eqref{definition of approximated solutions t}

\begin{eqnarray*}
& &\frac{\Delta_\Sigma\left(-\log(-\tilde{h})\right)-2c_1+\half s^2}{-s^2\tilde{h}}\\
&=& \frac{1}{-2\tilde{h}} - \frac{1}{s^2}\left[\frac{\Delta_\Sigma\left(-\log(-\tilde{h})\right)-2c_1{\tilde{h}}}{\tilde{h}}\right]
\end{eqnarray*}

\noindent is a smooth function function on $\Sigma$ at any $\tilde{\eta}\in \mathcal{U}$ and for any $R \in \mathbb{N}^m$,

\begin{eqnarray*}
& &\pRi \left[\frac{\Delta_\Sigma\left(-\log(-\tilde{h})\right)-2c_1+\half s^2}{-s^2\tilde{h}}\right] \\
&=&\pRi \left(\frac{1}{-2\tilde{h}}\right)-\frac{1}{s^2}\pRi \left[\frac{\Delta_\Sigma\left(-\log(-\tilde{h})\right)-2c_1{\tilde{h}}}{\tilde{h}}\right].
\end{eqnarray*}

\noindent Both terms in this expression are clearly smooth on $\Sigma$ and the factor $\frac{1}{s^2}$ of the second term, the only appearance of $s$, makes all its $z$-derivatives uniformly bounded, verifying \eqref{convergence 3}. For \eqref{convergence 4}, we simply observe that

\[-\tilde{h}e^{\tilde{v}_s} = \half + \frac{\Delta_\Sigma\left(-\log(-\tilde{h})\right)-2c_1}{-s^2} \to \half\]

\noindent and $\forall \alpha$ such that $1 \leq |\alpha| \leq |R|$,

\[\pRi \left(\tilde{h}e^{\tilde{v}_s}\right) = -\frac{\pRi \Delta_\Sigma \left(-\log\left(\tilde{h}\right)\right)}{s^2} \to 0\]

\noindent uniformly as $s\to \infty$. Constants are eliminated since $|\alpha| \geq 1$. These observations easily verify \eqref{convergence 4} and \eqref{convergence 5}. Lastly, we observe that the argument of $\log$ in \eqref{definition of the error term t}

\begin{eqnarray*}
& &\frac{\Delta_\Sigma \left(-\log\left(\tilde{h}\right)\right) -c_1 + \half s^2}{s^2} \\
&=&\half - \frac{1}{s^2} \left[\Delta_\Sigma \left(-\log\left(\tilde{h}\right)\right) -c_1\right]
\end{eqnarray*}

\noindent which clearly approaches $\half$ in all $H^{l,\infty}(\mathcal{U}\times \Sigma)$ as $s \to \infty$. It then follows that $\tilde{E_s} \to 0$ in all $H^{l,\infty}(\mathcal{U}\times \Sigma)$ and $\eqref{convergence 7}$ follows.

From \eqref{convergence 3}-\eqref{convergence 7}, we see that \eqref{convergence 1}-\eqref{convergence 6}, the three sufficient conditions for proving Proposition \ref{Precise Baptista's Conjecture}, are true if $\tilde{\varphi}_s$ is replaced by $\tilde{v}_s$. Therefore, it remains to show that at every $\tilde{\eta}\in \mathcal{U}$, the difference of $\pRi \tilde{\varphi}_s$ and $\pRi \tilde{v}_s$ converges to 0 in $H^{1,\infty}(\Sigma)$ for all $R\in \mathbb{N}^m$.

\begin{lemma}
For all multi-indices $R$, and $l \in \mathbb{N}$, we have

\[\lim_{s\to \infty} \Wnormloc{\p^R|_{\tilde{\eta}} \tilde{v}_s-\p^R|_{\tilde{\eta}} \tilde{\varphi}_s}{1}{p}{\Sigma}=0,\]

\noindent $\forall \tilde{\eta} \in \mathcal{U} \subset \Hol$. Here, $\p^R$ is the $R^{th}$ derivative with respect to coordinates $(w_1,\ldots,w_m)$ on $\mathcal{U}$.
\label{smooth convergence of function with approximations t}
\end{lemma}

\begin{proof}

We need to prove that for all $R \in \mathbb{N}^m$ and $\tilde{\eta}\in \mathcal{U}$,

\begin{equation}
\lim_{s\to\infty} \Lnormloc{\pRi \tilde{\varphi}_s - \pRi \tilde{v}_s}{\infty}{\Sigma}=0
\label{decay of difference}
\end{equation}

\noindent and

\begin{equation}
\lim_{s\to\infty} \Lnormloc{\pRi \frac{\p\tilde{\varphi}_s}{\p z} - \pRi \frac{\p \tilde{v}_s}{\p z}}{\infty}{\Sigma}=0.
\label{decay of difference derivative}
\end{equation}

The proof is essentially a repetition of that of Lemma \ref{smooth convergence of function with approximations}. We start with the difference of Laplacians of the quantities we wish to bound:

\begin{equation}
\Delta_\Sigma \left(\tilde{\varphi}_s-\tilde{v}_s\right)=-s^2\tilde{h}e^{\tilde{\varphi}_s}+s^2\tilde{h}e^{\tilde{v}_s}-\tilde{E}_s
\label{difference of Laplacians}
\end{equation}

\noindent and

\begin{eqnarray}
& &\Delta_\Sigma \left(\frac{\p\tilde{\varphi}_s}{\p z}-\frac{\p \tilde{v}_s}{\p z}\right)\nonumber \\
&=&-s^2 \frac{\p \tilde{h}}{\p z}e^{\tilde{\varphi}_s}+s^2 \frac{\p \tilde{h}}{\p z}e^{\tilde{v}_s}-\frac{\p \tilde{E}_s}{\p z} \nonumber \\
& &+s^2\tilde{h}\frac{\p \tilde{v}_s}{\p z}e^{\tilde{v}_s}-s^2\tilde{h}\frac{\p \tilde{\varphi}_s}{\p z}e^{\tilde{v}_s} + Q(z) \left(\frac{\p \tilde{\varphi}_s}{\p z} - \frac{\p \tilde{v}_s}{\p z} \right), \nonumber \\
\label{difference of Laplacians of derivatives}
\end{eqnarray}

\noindent where $Q(z)$ is a smooth function on $\Sigma$ arisen from the Riemannaian curvature tensors on $\Sigma$ and their derivatives. It is in particular independent of coordinates of $\mathcal{U}$. Similar to the proof of Lemma \ref{smooth convergence of function with approximations}, we apply $\pRi$ to \eqref{difference of Laplacians} and \eqref{difference of Laplacians of derivatives} above, follow by induction on $|R|$ and arguments from the maximum principle.

For $|R|=0$, no derivative on coordinates of $\mathcal{U}$ is taken. \eqref{decay of difference} and \eqref{decay of difference derivative} are merely special cases of Lemma \ref{smooth convergence of function with approximations} with $l=1$, as the holomorphic map chosen there is arbitrary as well. Suppose that \eqref{decay of difference} and \eqref{decay of difference derivative} hold for $|R|\leq l$. The inductive step, as in the proof of Lemma \ref{smooth convergence of function with approximations}, is established from the following crucial claim:

\begin{equation}
\lim_{s\to\infty}\Lnormloc{s^2\left(e^{\tilde{\varphi}_s(\tilde{\eta},z)}-e^{\tilde{v}_s(\tilde{\eta},z)}\right)}{\infty}{\Sigma}=0\hspace{0.5cm}\forall \tilde{\eta}\in\mathcal{U}.
\label{crucial claim 2}
\end{equation}

\noindent This is simply a claim that the Claim \eqref{convergence of varphi s and v s} holds for every smooth function $\tilde{\varphi}_s(\tilde{\eta},z)$ and $\tilde{v}_s(\tilde{\eta},z)$ induced from $\tilde{\eta} \in \mathcal{U}$, which is indeed true as the smooth functions in section 3 are all induced from an arbitrary holomorphic map.

\eqref{decay of difference} is almost an identical statement to Lemma \ref{smooth convergence of function with approximations} with the multi-index $I$ replaced by $R$. That is, the derivatives are taken with respect to coordinates of $\mathcal{U}$ instead of $\Sigma$. This replacement actually simplifies the computation considerably as $\p^R$ and $\Delta_\Sigma$ are independently defined on the separate components of $\mathcal{U}\times\Sigma$ and therefore commute. Consequentially, the curvature terms $Q^j(Rm)$'s in the proof of Lemma \ref{smooth convergence of function with approximations} do not appear here when commuting $\p^R$ and $\Delta_\Sigma$. With this liberty at hand, we readily compute

\begin{eqnarray}
& &\Delta_\Sigma \left(\pRi \tilde{\varphi}_s - \pRi \tilde{v}_s\right) \nonumber \\
=&-& s^2 \tilde{h}\left(\pRi \tilde{\varphi}_s - \pRi \tilde{v}_s\right) + \sum_{j\in \{R\}\cup M^l}\left(\tilde{A}_{j,s}+\tilde{B}_{j,s}\right)|_{(\tilde{\eta},z)} \nonumber \\
&+&\tilde{C}_s(\tilde{\eta},z)-\pRi \tilde{E}_s.
\label{difference tilde Laplacian}
\end{eqnarray}

\noindent Here, the smooth functions $\tilde{A}_{j,s},\tilde{B}_{j,s}$, $\tilde{C}_s$, and index set $M^l$ are all defined identically as $A_{j,s}$, $B_{j,s}$, $C_s$, and $M^l$ in the proof of Lemma \ref{smooth convergence of function with approximations}, with $h$, $v_s$, $\varphi_s$, and $Q^j(Rm)$'s replaced by $\tilde{h}$, $\tilde{v_s}$, $\tilde{\varphi_s}$, and $0$, respectively. Claim \eqref{crucial claim 2} and inductive hypothesis are then applied identically to obtain the follow decay conditions:

\begin{equation}
\lim_{s\to\infty}\Lnormloc{\frac{\tilde{A}_{j,s}(\tilde{\eta},z)}{s^2}}{\infty}{\Sigma}=\lim_{s\to\infty}\Lnormloc{\frac{\tilde{B}_{j,s} (\tilde{\eta},z)}{s^2}}{\infty}{\Sigma}
=\lim_{s\to\infty}\Lnormloc{\frac{\tilde{C}_s(\tilde{\eta},z)}{s^2}}{\infty}{\Sigma}=0
\label{decay of A B C tilde}
\end{equation}

\noindent Maximum principle is then identically applied. Namely, for each $\tilde{\eta}$ and $s$, there exist $\xs, \ys \in \Sigma$ such that for all $z\in\Sigma$,

\begin{eqnarray}
 & &\pRi \tilde{\varphi}_s-\pRi \tilde{v}_s \nonumber \\
&\leq& \frac{e^{-\tilde{\varphi}_s}}{\tilde{h}}
\left(\sum_{j\in \M}\left[\frac{\tilde{A}_{j,s}}{s^2}+\frac{\tilde{B}_{j,s}}{s^2}\right]+\frac{\tilde{C}_s}{s^2}-\frac{\p^R \tilde{E}_s}{s^2}\right)\vline_{(\tilde{\eta},\xs)}, \nonumber \\
\label{upper bound tilde}
\end{eqnarray}

\noindent and

\begin{eqnarray}
 & &\pRi \tilde{\varphi}_s-\pRi \tilde{v}_s \nonumber \\
&\geq& \frac{e^{-\tilde{\varphi_s}}}{\tilde{h}}
\left(\sum_{j\in \M}\left[\frac{\tilde{A}_{j,s}}{s^2}+\frac{\tilde{B}_{j,s}}{s^2}\right]+\frac{\tilde{C}_s}{s^2}-\frac{\p^R \tilde{E}_s}{s^2}\right)\vline_{(\tilde{\eta},\ys)}. \nonumber \\
\label{lower bound tilde}
\end{eqnarray}

\noindent \eqref{convergence 7} and \eqref{decay of A B C tilde} then imply that the right hand sides of \eqref{upper bound tilde} and \eqref{lower bound tilde} decay to 0 uniformly as $s\to\infty$, verifying \eqref{decay of difference}.

The uniform decay \eqref{decay of difference derivative} is obtained similarly despite its more tedious and lengthy computations. With the case $|R|=0$ verified and $|R|\leq l$ assumed, we aim to prove \eqref{decay of difference derivative} for an arbitrary $R\in \mathbb{N}^m$ with $|R|=l+1$. Applying $\pRi$ to both sides of \eqref{difference of Laplacians of derivatives} we obtain

\begin{eqnarray}
& &\Delta_\Sigma \left(\pRi \frac{\p\tilde{\varphi}_s}{\p z}-\pRi \frac{\p \tilde{v}_s}{\p z}\right) \nonumber \\
&=&-\tilde{h}\left[\left(\pRi \frac{\p \tilde{\varphi}_s}{\p z}\right) s^2e^{\tilde{\varphi}_s}-\left(\pRi \frac{\p \tilde{v}_s}{\p z}s^2e^{\tilde{v}_s}\right)\right]\vline_{(\tilde{\eta},z)} \nonumber \\ & &-s^2\tilde{h}e^{\tilde{\varphi}_s}\left[2-\frac{Q(z)}{s^2\tilde{h}e^{\tilde{\varphi}_s}}\right] \left(\p^R \frac{\p \tilde{\varphi}_s}{\p z} -\p^R \frac{\p \tilde{v}_s}{\p z}\right)\vline_{(\tilde{\eta},z)} \nonumber \\
& &+\sum_{j\in \{R\} \cup M^l} \hat{A}_{j,s}(\tilde{\eta},z)+\sum_{j\in \cup M^l} \hat{B}_{j,s}(\tilde{\eta},z)+\hat{C}_s(\tilde{\eta},z)+\hat{D}_s(\tilde{\eta},z)
-\tilde{h}\hat{\rho}(s) \p^R \frac{\p \tilde{v}_s}{\p z}\vline_{(\tilde{\eta},z)}, \nonumber \\
\label{Laplacian difference again}
\end{eqnarray}

\noindent where

\begin{eqnarray}
\hat{A}_{j,s}=\sum_{m^j(t)\neq ((j),(1))}& &\left[a_{m^j(t)}\p^{R-j}\left(\tilde{h}\frac{\p \tilde{v}_s}{\p z}+\frac{\p \tilde{h}}{\p z}\right)\left(\p^{m_i} \frac{\p \tilde{v}_s}{\p z}\right)^{t_i}\right]s^2e^{\tilde{v}_s} \nonumber \\
&-&\left[a_{m^j(t)}\p^{R-j}\left(\tilde{h}\frac{\p \tilde{\varphi}_s}{\p z}+\frac{\p \tilde{h}}{\p z}\right)\left(\p^{m_i} \frac{\p\tilde{\varphi}_s}{\p z}\right)^{t_i}\right]s^2e^{\tilde{\varphi}_s}, \nonumber \\
\label{definition of A j s hat}
\end{eqnarray}

\begin{equation}
\hat{B}_{j,s}=\sum_{j\in M^l} a_{((j),(1))}\left[\p^{R-j}\tilde{h}\left(\p^j \frac{\p \tilde{v}_s}{\p z}\right)\right]s^2e^{\tilde{v}_s}-a_{((j),(1))}\left[\p^{R-j}\tilde{h}\left(\p^j \frac{\p \tilde{\varphi}_s}{\p z}\right)\right]s^2e^{\tilde{\varphi}_s},
\label{definition of B j s hat}
\end{equation}

\begin{equation}
 \hat{C}_s=-s^2\tilde{h}e^{\tilde{\varphi}_s}\ \left(1-e^{\tilde{v}_s-\tilde{\varphi}_s}\right)\p^R \left(\frac{\p \tilde{v}_s}{\p z}\right),
 \label{definition of C  s hat}
\end{equation}

\begin{equation}
\hat{D}_s=-\frac{\p \tilde{h}}{\p z} \left[\left(\p^R \tilde{\varphi}_s\right)s^2 e^{\tilde{\varphi}_s}-\left(\p^R \tilde{v}_s\right)s^2 e^{\tilde{v}_s}\right],
\label{definition of D s hat}
\end{equation}

\noindent and

\begin{equation}
\hat{\rho}(s)\to 0 \text{ in }L^\infty(\Sigma) \text{ as } s \to \infty.
\label{definition of rho s}
\end{equation}

\noindent The inductive hypothesis, \eqref{convergence 3}, and \eqref{crucial claim 2} again form the required decaying conditions on all the functions on the last line of \eqref{Laplacian difference again} for us to apply the maximum principle:

\begin{eqnarray}
& &\lim_{s\to\infty}\Lnormloc{\frac{\hat{A}_{j,s}(\tilde{\eta},z)}{s^2}}{\infty}{\Sigma}=\lim_{s\to\infty}\Lnormloc{\frac{\hat{B}_{j,s}(\tilde{\eta},z)}{s^2}}{\infty}{\Sigma}
=\lim_{s\to\infty}\Lnormloc{\frac{\hat{C}_{s}(\tilde{\eta},z)}{s^2}}{\infty}{\Sigma} \nonumber \\
&=&\lim_{s\to\infty}\Lnormloc{\frac{\hat{D}_{s}(\tilde{\eta},z)}{s^2}}{\infty}{\Sigma}
=\lim_{s\to\infty}\Lnormloc{\frac{\tilde{h}\hat{\rho}(s) \p^R \frac{\p \tilde{v}_s}{\p z}}{s^2}\vline_{(\tilde{\eta},z)}}{\infty}{\Sigma}=0 \nonumber \\
\label{decay of A B C D hat}
\end{eqnarray}

\noindent for all $\tilde{\eta}\in \mathcal{U}$. The maximum principle of $\Delta_\Sigma$ then guarantees the existences of $\xs,\ys \in \Sigma$ so that for all $z\in\Sigma$, we have

\begin{eqnarray}
& &\pRi \frac{\p \tilde{\varphi_s}}{\p z} - \pRi \frac{\p \tilde{v_s}}{\p z} \nonumber \\
&\leq& \left(\frac{e^{-\tilde{\varphi_s}}}{-\tilde{h}\left[2-\frac{Q}{s^2\tilde{h}e^{\tilde{\varphi_s}}}\right]}
\left[\sum_{j\in\{R\}\cup M^l} \frac{\hat{A}_{j,s}}{s^2}+\sum_{j \in M^l} \frac{\hat{B}_{j,s}}{s^2}+\frac{\hat{C}_s}{s^2}+\frac{\hat{D}_s}{s^2}-\frac{\tilde{h}\rho(s)\p^R\left(\frac{\p \tilde{v_s}}{s^2}\right)}{s^2}\right]\right)\vline_{(\tilde{\eta},\xs)} \nonumber \\
\label{last max}
\end{eqnarray}

\noindent and

\begin{eqnarray}
& &\pRi \frac{\p \tilde{\varphi_s}}{\p z} - \pRi \frac{\p \tilde{v_s}}{\p z} \nonumber \\
&\geq& \left(\frac{e^{-\tilde{\varphi_s}}}{-\tilde{h}\left[2-\frac{Q}{s^2\tilde{h}e^{\tilde{\varphi_s}}}\right]}
\left[\sum_{j\in\{R\}\cup M^l} \frac{\hat{A}_{j,s}}{s^2}+\sum_{j \in M^l} \frac{\hat{B}_{j,s}}{s^2}+\frac{\hat{C}_s}{s^2}+\frac{\hat{D}_s}{s^2}-\frac{\tilde{h}\rho(s)\p^R\left(\frac{\p \tilde{v_s}}{\p z}\right)}{s^2}\right]\right)\vline_{(\tilde{\eta},\ys)}. \nonumber \\
\label{last min}
\end{eqnarray}

\noindent It then follows from \eqref{decay of A B C D hat} that the right hand sides of \eqref{last max} and \eqref{last min} decay to 0 in $L^\infty (\Sigma)$, proving our remaining claim \eqref{decay of difference derivative}.

\end{proof}

We have proved, that on the coordinate patch $\mathcal{U} \subset \Hol$, the function

\[\p^R F_s^i = \p^R g_s^* \left(\frac{\p}{\p w_i},\frac{\p}{\p w_i}\right)\]

\noindent converges pointwise to the smooth function

\[\int_\Sigma \frac{1}{2} \left[\p^R \left< \frac{\p}{\p w_i},\frac{\p}{\p w_i}\right>_{H_{FS}}\right]=\frac{1}{2} \p^R \int_\Sigma \left< \frac{\p}{\p w_i},\frac{\p}{\p w_i}\right>_{H_{FS}},\]

\noindent for all multi-index $R \in \mathbb{N}^m$ as $s \to \infty$. All functions $\p^R F_s^i$ and the limiting function are bounded on $\mathcal{U}$ and therefore admit smooth extensions to the compact set $\bar{\mathcal{U}}$. Since the limiting function is smooth, it follows that the functions $\p^R F_s^i$ converge uniformly to

\[\frac{1}{2} \p^R \int_\Sigma \left< \frac{\p}{\p w_i},\frac{\p}{\p w_i}\right>_{H_{FS}},\]

\noindent on $\bar{\mathcal{U}}$ which proves the smooth convergence of $g_s^*$ to a multiple of $<\cdot,\cdot>_{L^2}$ on $\mathcal{U}$. This completes the proof of Proposition \ref{Precise Baptista's Conjecture}.

\end{proof}

Proposition \ref{Precise Baptista's Conjecture} provides a plausible approach to prove Conjecture 5.3 in \cite{Ba}, which conjectures a formula of the volume of $\Hol$ with respect to $<\cdot,\cdot>_{L^2}$. The volume of $\nu_k(s)$ with respect to $g_s$ has been explicitly computed \cite{Ba} (See Theorem 5.1 there). In the notations we use in this paper, the formula is

\begin{equation}
Vol \nu_k(s) = \sum_{i=0}^b \frac{b!k^{b-i}}{i!(q-i)!(b-i)!}\left(\frac{4\pi}{s^2}\right)^i\left(Vol \Sigma - \frac{4\pi}{s^2}r\right)^{q-i},
\label{volume of moduli space}
\end{equation}

\noindent where

\[q=b+k(r+1-b)-1.\]

\noindent \eqref{volume of moduli space} is off by a factor of $\pi^q$ from \cite{Ba}, as we adopt the normalized K\"ahler form $\omega_{FS}$ here. Also, we have $4\pi$ here, instead of $2\pi$, as the adiabatic parameter $s^2$ here corresponds to $2e^2$ in \cite{Ba}. Since $\nu_{k,0}(s)$ is an open dense subset of $\nu_k(s)$, for $s<\infty$, \eqref{volume of moduli space} is also a formula $Vol \nu_{k,0}(s)$. Since $\Phi_\infty : \nu_{k,0}(\infty) \to \Hol$ is an isometry by Proposition \eqref{Precise Baptista's Conjecture}, letting $s\to\infty$ in \eqref{volume of moduli space} formally yields a conjectural formula for the volume of $\Hol$:

\begin{equation}
Vol \Hol = \frac{k^b}{q!}\left( Vol \Sigma\right)^q.
\label{volune of the space of Hol}
\end{equation}

\noindent This formula has been verified in \cite{Sp} for the case $b=0$ and $r=1$ using entirely independent techniques that are quite special to this given case. It is valid in general if \eqref{volume of moduli space} is true for $s=\infty$. One however, needs to confirm that the $L^2$ volume of $\nu_k(s)$ does not concentrate around $\nu_k(s) - \nu_{k,0}(s)$ so that \eqref{volune of the space of Hol} is equal to $\lim_{s\to \infty} \nu_k(s)$. The affirmation is not immediate, as sketched in the next section, that analytic defect appears on sections with common zeros which is also exhibited by the loss of topological invariants. It is author's great interest to confirm that the singularities of $L^2$ metrics on $\nu_k(s)$ produced as $s \to \infty$ does not impact the continuity of \eqref{volume of moduli space} in $s$ and the plausible argument above is indeed valid.

\section{Failure of the Results from Common Zeros and Bubbling}
We have restricted our discussion to the open subset $\nu_{k,0}(s)$ of $\nu_k(s)$ where sections do not
vanish simultaneously. This leads to the non-vanishing of the function $h$, allowing us to pick the
constant $K \in \mathbb{R}$ to control the smooth function $\Delta(-\log(-h))$ (See the proof of
Theorem \ref{Kazdan Warner Analysis} and Theorem \ref{Main Theorem}). When $k$ sections do have common
zeros, the function $h$ vanishes at the common zeros and the function $\Delta(-\log(-h))$ is no longer
smooth and bounded. It is therefore no longer possible to pick such $K$ to bound $\varphi_{+,s}$ and
$\varphi_{-,s}$. Without this vital condition, Main Theorem \ref{Main Theorem} does not hold and it is not
possible to obtain convergent behaviors of the functions $u_s$.

Although it is still possible to obtain the super and sub solutions for each $s$ in the proof of
Theorem \ref{Kazdan Warner Analysis} following the choices of $\varphi_{+,s}$ and $\varphi_{-,s}$ in
\cite{kw}, where $s=1$, these functions will not be uniformly bounded. Their $L^\infty$ norms
grow like $s^2$, failing to satisfy the crucial condition of the Main Theorem on the uniform bounds of
$\varphi_s$.

In fact, when sections do have common zeros, convergence of the family of solutions of vortex equations
\eqref{s-vortex} to those of \eqref{infinity-vortex} contradicts the topological constraint of the line
bundle $L$. An easy example can be observed for single section vortices $k=1$. At $s=\infty$, equation
\eqref{infinity-vortex} indicates that the section never vanishes on $\Sigma$, which is impossible for
line bundle of positive degree. However, as we have seen from the constructions in section 3, values of $s$ correspond to various gauge classes of connections and sections, which do not alter the topological structure of $L$. Analytically, the equation for
$\varphi_\infty$, namely $h e^{\varphi_\infty} +c_2=0$, can never be true unless $h$ contains singular
points. Consequentially, the density for Yang-Mills-Higgs functional is expected to blow up at the common
zeros of the sections, even though the energy functional stays bounded. One can certainly remedy this
setback by defining some smooth extension of the vortices across the singularities. However, it is then
necessary to sacrifice some topological data form our initial setting. This phenomenon is known as the
"bubbling" of vortices. Descriptions of the bubbles, as well as the leftover bundles, have been thoroughly
described in \cite{CGRS}, \cite{O}, \cite{W}, \cite{X}, and \cite{Z} in more general settings of symplectic
vortex equations.

\section{Remarks on Possible Generalizations}
At the time of submission of this article, a more generalized version of Baptista's conjecture is posed in \cite{Ba1}. The conjecture is similar, except that the Riemann surface $\Sigma$ is replaced by an arbitrary compact K\"ahler manifold and $\pk$ is replaced with a toric manifold. The analog of $\Hol$ there (with naturally defined $L^2$ metric) holomorphically embeds into the analog of $\nu_k(s)$ there, and it is conjectured that as $s \to \infty$ the embedding is isometric. As our Main Theorem does not restrict the dimension of the manifold, it is then natural to attempt to generalize our results to this setting.

Another possible generalization is to allow certain singularities on the Hermitian metrics. In \cite{D}, several regularity results are available for the types of elliptic PDE's considered in section 3, with background metric possessing conic singularities. It suggests possible generalization to our Main Theorem for the negative function $h$ with conic singularities. Such a result might possibly provide a more precise analytic picture on the bubbling phenomenon.

The author is eager to explore any possibility toward these two directions of generalizations.

\vspace{0.5cm}
\begin{center}
ACKNOWLEDGEMENT
\end{center}
\vspace{0.5cm}
\small{This paper is a part of the author's Ph.D thesis in the University of Illinois at Urbana-Champaign, USA. The
author wishes to express sincere gratitude toward his thesis advisor, Steven Bradlow, and co-advisor Gabriele
La Nave, for their guidance and invaluable inspirations. The author is also grateful toward Eduard-Wilhelm Kirr, for
the insightful discussions on the analytic details of the Main Theorem. Last but not the least, the author thanks
the hospitality of the Center for Mathematics and Theoretical Physics of National Central University in Chung-Li,
Taiwan, hosted by M.K. Hong, C.H. Hsu,  and S.Y. Yang. Major constructions of this paper were completed during the visit to the center.}


\begin{thebibliography} {Library}



\bibitem [B]{Ba}
J.M. Baptista,
{\em On the $L^2$ Metrics of Vortex Moduli Spaces}, Nuclear Physics B, 844, 308-333 (2010).

\bibitem [B1]{Ba1}
J.M. Baptista,
{\em Moduli Spaces of Abelian Vortices on K\"ahler Manifolds}, arXiv: 1211.0012.


\bibitem [Br] {Br}
S.B. Bradlow,
{\em Vortices in Holomorphic Line Bundles over Closed K\"ahler Manifolds}, Commun. Math. Phys. 135, 1-17 (1990).

\bibitem [Br1] {Br1}
S.B. Bradlow, {\em Special Metrics and Stability for Holomorphic Bundles with Global Sections}, J. Diff. Geom. 33, 169-213 (1991).

\bibitem [B-D-W]{BDW}
A. Bertram, G. Daskalopoulos, and R. Wentworth,
{\em Gromov Invariants for Holomorphis Maps from Riemann Surfaces to Grassmannians}, Journal of the American Mathematical Society. 9, 529-571 (1996).

\bibitem [C-G-R-S] {CGRS} K. Cieliebak, A.R. Gaio, I. Mundet i Riera, D.A. Salamon,
{\em The Symplectic Vortex Equaions and Invariants of Hamiltonian Group Actions}, J. Symplectic Geom. 1 (2002), 3, 543-645.


\bibitem [D] {D}
S.Donaldson,  {\em K\"ahler metrics with cone singularities along a divisor,} arxiv:1102.1196

\bibitem [D-K] {DK} S. K. Donaldson, P.B. Kronheimer, {\em The Geometry of Four-Manifolds}, Oxford Science Publications, (1990).

\bibitem [G] {G}
O. Garcia-Prada,
{\em A Direct Existence Proof for the Vortex Equations over a Riemann Surface}, Bull. London Math Soc. 26(1), 88-96, (1994).

\bibitem [G-H] {GH}
P.A. Grirriths, J. Harris,
{\em Principles of Algebraic Geometry}, John Wiley And Sons, Inc., (1978).

\bibitem [G-S] {GS}
A. Gaio, D. Salamon,
{\em Gromov-Witten Invariants of Symplectic Quotients and Adiabatic Limits}, J. Symplectic Geom. 3,  1, 55-159, (2005).

\bibitem [Gr] {Gr}
P.A. Griffiths, {\em Introduction to Algebraic Curves}, American Mathematical Society, Vol 76, (1983).

\bibitem [J-T] {JT}
A. Jaffe, C. Taubes,
{\em Vortices and Monopoles}, Birkh\"auser, (1981).

\bibitem [K]{K}
S. Kobayashi, {\em Differential Geometry of Complex Vector Bundles}, Iwanami Shoten, Publishes and Princeton University Press, 1987.

\bibitem [K-M] {KM}
S. Kallel, J. Milgram, {\em Space of Holomorphic Maps}, J. Diff. Geom. 47, 321-375 (1997).

\bibitem [K-W]{kw}
J. Kazdan, F.W. Warner,
{\em Curvature Functions for Compact ~2-Manifolds}, Ann. Math 2, 99, 14-47 (1978).

\bibitem [M] {M}
N.S. Manton,
{\em A Remark on the Scattering of BPS Monopoles}, Phys. Lett. 110B, 54-56 (1982).

\bibitem [Mi] {Mi}
R. Miranda,
{\em Algebraic Curves and Riemann Surfaces}, American Mathematical Society, Vol 5, 1995.

\bibitem [M-P]{MP}
D. Morrison, M. Plesser,
{\em Summing the Instantons: Quantum Cohomology and Mirror Symmetry in Toric Varieties}, Nuclear Physics B, 440, 279-354 (1995).

\bibitem [O] {O}
A. Ott,
{\em Removal of Singularities and Gromov Compactness for Symplectic Vortices}, arXiv: 0912.2500.

\bibitem [R] {Ro}
N. Romao,
{\em Gauged Vortices in a Background}, J. Phys. A: Math. Gen. 38 9127 (2005).

\bibitem [S] {S}
M. Samols,
{\em Vortices in Holomorphic Line Bundles over Closed K\"ahler Manifolds}, Commun. Math. Phys. 135, 1-17 (1990).

\bibitem [Sp] {Sp}
J.M. Speight,
{\em The Volume of the Space of Holomorphic Maps From $\mathbb{S}^2$ to $\pk$}, J. Geom. and Phys., 61, 77-84 (2011).

\bibitem [V] {V}
J.A. Viaclovsky, {\em Math 865, Topics in Riemannian Geometry}, Fall 2007 Class Notes in University of Wisconsin, Madison.

\bibitem [W] {Wi}
E. Witten,
{\em Phases of $N=2$ Theories in Two Dimensions}, Nuclear Physics B, 403, 159-222 (1993).

\bibitem [Wo] {W}
C. Wodward,
{\em Quantum Kirwan Morphism and Gromov-Witten Invarants of Quotients}, arXiv: 1204.1765, April 2012.

\bibitem [X] {X}
G. Xu,
{\em $U(1)$- Vortices and Quantum Kirwan Map}, arXiv: 1202.4096.

\bibitem [Z] {Z}
F. Ziltener,
{\em A Quantem Kirwan Map: Bubbling and Fredholm Theory}, Memoirs of the American Mathematical Society (2012).

\end{thebibliography}
\end{document}